%% file: main.tex
\documentclass[12pt]{iopart}

\usepackage{iopams}  
\usepackage{graphicx}

\usepackage{enumitem}
\expandafter\let\csname equation*\endcsname\relax

\expandafter\let\csname endequation*\endcsname\relax
\usepackage{amsmath}
\usepackage{amsthm}
\usepackage[table]{xcolor}
\usepackage{subcaption}
\usepackage[lined,boxed,commentsnumbered, ruled,vlined,linesnumbered, english]{algorithm2e}
\usepackage{booktabs}
\usepackage{hyperref}
\usepackage{multirow}
\usepackage{floatrow}
\usepackage{makecell,tabularx}
\usepackage{rotating}
\usepackage[export]{adjustbox}
\usepackage{cleveref}

\input{style}   
\input{shortcuts}

\begin{document}

\title[Novel approach for multiphoton microscopy imaging]{A Novel Variational Approach for Multiphoton Microscopy Image Restoration: from PSF Estimation to 3D Deconvolution}

\author{Julien Ajdenbaum$^1$, Emilie Chouzenoux$^1$, Claire Lefort$^2$, \\S\'egol\`ene Martin$^1$$^*$, Jean-Christophe Pesquet$^1$}

\address{$^1$ Centre de Vision Num\'erique, Universit\'e Paris-Saclay, Inria, CentraleSup\'elec, Gif-sur-Yvette, France \\ 
$^2$ XLIM Research Institute, UMR CNRS 7252, Universit\'e de Limoges, Limoges, France\\
$^*$ Corresponding author}

\begin{abstract}
 In multi-photon microscopy (MPM), a recent \emph{in-vivo} fluorescence microscopy system, the task of image restoration can be decomposed into two interlinked inverse problems: firstly, the characterization of the Point Spread Function (PSF) and subsequently, the deconvolution (i.e., deblurring) to remove the PSF effect, and reduce noise. 

The acquired MPM image quality is critically affected by PSF blurring and intense noise. The PSF in MPM is highly spread in 3D and is not well characterized, presenting high variability with respect to the observed objects. This makes the restoration of MPM images challenging. Common PSF estimation methods in fluorescence microscopy, including MPM, involve capturing images of sub-resolution beads, followed by quantifying the resulting ellipsoidal 3D spot. In this work, we revisit this approach, coping with its inherent limitations in terms of accuracy and practicality.
We estimate the PSF from the observation of relatively large beads (approximately $1\mu \text{m}$ in diameter). This goes through the formulation and resolution of an original non-convex minimization problem, for which we propose a proximal alternating method along with convergence guarantees.

Following the PSF estimation step, we then introduce an innovative strategy
to deal with the high level multiplicative noise degrading the acquisitions. We rely on a heteroscedastic noise model for which we estimate the parameters. We then solve a constrained optimization problem to restore the image, accounting for the estimated PSF and noise, while allowing a minimal hyper-parameter tuning. Theoretical guarantees are given for the restoration algorithm.

These algorithmic contributions lead to an end-to-end pipeline for 3D image restoration in MPM, that we share as a publicly available Python software. We demonstrate its effectiveness through several experiments on both simulated and real data. 
\end{abstract}

\section{Introduction}

Multi-photon microscopy (MPM) is a non-invasive laser imaging technique that selectively induces fluorescence in a thin plane through localized nonlinear excitation while avoiding excitation elsewhere \cite{diaspro2006multi}. This technique enables three-dimensional imaging with infrared light, reaching depths that are at least double compared to traditional single-photon microscopy \cite{gobel2007imaging,centonze1998multiphoton,larson2011multiphoton}. Particularly, MPM contactless and non-invasive nature makes it suitable for imaging diverse objects, spanning from living organisms \cite{zipfel2003nonlinear} to materials \cite{jonard2022multiphoton}. 

Though MPM offers significant advantages, its widespread adoption remains limited due to degradations affecting the images, especially blur and noise, requiring innovative computational solutions to unlock its full potential.
Central to MPM imaging is the Point Spread Function (PSF). This function describes how the imaging system responds to a a Dirac impulse, in practice a single-point light object. In practical terms, this response appears as an extended spot in the image around the original object location, which is responsable for distortions and blurry aspects \cite{young2011effectsa,young2011effectsb,descloux2016aberrations}, particularly in the depth direction. As MPM typically operates as a non-coherent and space-invariant system, it is commonplace to model the captured observation as a convolution of the actual object with a blur kernel. In the following, we will use the terms \emph{PSF} and \emph{blur kernel} interchangeably as they refer to the same mathematical object. A major challenge with the PSF in MPM is its strong dependency on the medium \cite{dong2003characterizing}, implying that it should be individually estimated for each new acquisition. 

In MPM, the restoration process aims to recover the sought (non-degraded) object of interest, $\bar{x}\colon \R^3 \longrightarrow \R$, jointly with an estimation of the Point Spread Function (PSF), $\bar{h}\colon \R^3 \longrightarrow \R$, guided by the following linear observational model, at each location $\bs{u}$ of the 3D volume:
\begin{equation}\label{eq:blind_inverse_pb}
(\forall \bs{u}\in \R^3) \quad y(\bs{u}) = \mathcal{D}(\bar{h} \star \bar{x} + \alpha)(\bs{u}).
\end{equation}
In \eqref{eq:blind_inverse_pb}, $y\colon \R^3 \longrightarrow \R$ stands for the observation, $\alpha \in \R$ represents a background factor, and $\star$ the continuous 3D convolution operation, here to be assumed to be defined. Additionally, $\mathcal{D}$ refers to the noise model. In the computational imaging literature, there are two main ways of addressing the above inverse problem. Some methods directly tackle the \emph{blind} inverse problem represented by \eqref{eq:blind_inverse_pb}, i.e. they simultaneously estimate both $\bar{x}$ and $\bar{h}$ \cite{vorontsov2017new,prato2013convergent,huang2022unrolled,debarnot2021deepblur}. However, this approach can be computationally demanding due to its non-convex nature and potential underdetermined nature.
Another strategy is to break down the problem into two successive inverse problems. In this two-step approach, one first performs a calibration step \cite{krist2011years,hogbom1974aperture,dusch2007three,huang2021probabilistic,samuylov2019modeling}, providing an estimate of the PSF $\bar{h}$ in a controlled situation where $\bar{x}$ is a known (i.e., calibrated) entity. Then, utilizing the estimated PSF, one can infer an estimate of another, more complex, unknown object $\bar{x}$.

PSF calibration in MPM typically involves imaging sub-resolution fluorescent beads in a homogeneous medium \cite{cole2011measuring,doi2018high} or directly in the sample (\emph{in situ}) \cite{lefort2021famous,doi2018high,tiedemann2006image}. The image of such a bead captured with the MPM device appears as a spread spot, which is then fit with a shape function (e.g., Gaussian) to estimate the PSF characteristics \cite{tiedemann2006image,zhang2007gaussian,guo2011simple,chouzenoux2019optimal}. However, this method presents significant practical limitations. Firstly, handling beads as small as $0.2 \mu \text{m}$ in diameter requires special care, as sub-micrometric objects are difficult to even observe with the microscope. As a result, one often injects hundreds of beads into the sample, often leading to beads aggregation that compromises the accurate completion of the PSF calibration protocol. Secondly, observing the impulse response of the instrument on such tiny beads requires the use of high laser power, which can potentially damage the sample itself. These limitations highlight the need for an alternative and a more practical protocol for PSF estimation in MPM.

Following PSF estimation, the subsequent task is image deblurring and denoising.
Over the years, several computational solutions have been proposed to solve this inverse problem \cite{dao2015model,strohl2015joint,sarder2006deconvolution} in the context of microscopy. A common strategy is to minimize a cost function which balances a data-fitting term with a regularization term \cite{difato2004improvement,monvel2003image,dey2006richardson}. Yet, such an approach often suffers from a time-intensive parameter tuning phase, made even more challenging by subjective decision-making in the absence of a clear ground-truth. Further difficulties emerge from consideration on the nature of the noise in MPM. While most studies assume a standard Gaussian additive noise \cite{doi2018high,danielyan2014denoising}, the noise in MPM images is more appropriately considered as multiplicative, as pointed out in \cite{crivaro2011multiphoton,niu2022boundary}. However, such a noise model can be tedious to deal with numerically at large-scale, often involving minimizing a non-smooth data-fidelity term \cite{zanella2009efficient,chouzenoux2015convex}. While this issue has been tackled in confocal microscopy \cite{vicidomini2009application,zanella2009efficient,bohra2019variance}, its exploration in MPM imaging remains scarce.

In this work, we present a novel approach for addressing the MPM image restoration inverse problem in an end-to-end fashion.
Our comprehensive restoration pipeline revisits the conventional restoration protocol from acquisition to the final restored outcome. Within this pipeline, we introduce two major advancements:

\begin{enumerate}
\item \textbf{PSF Calibration Step.} Our contributions regarding the estimation of the instrumental PSF are two-fold. 
\begin{itemize}
\item We introduce a novel optimization problem formulation for estimating the a Gaussian-shaped PSF assuming large diameter beads as opposed to sub-resolution beads. Our selected bead size is small enough to accurately capture small-scale aberrations and system imperfections, while effectively mitigating the practical issues associated with sub-resolution beads.
\item Next, we address the PSF estimation inverse problem by deploying an alternating proximal approach whose convergence is proved.
\end{itemize}
\item \textbf{Image restoration Step.}
Our contributions to the image restoration phase encompass two key elements. 
\begin{itemize}
\item 
First, we adopt an additive heteroscedastic noise model that approximates the Poisson-Gaussian noise. We then introduce a quantization-based image analysis method to estimate the noise model parameters.
\item Additionally, we address the parameter tuning issue raised in penalized restoration approaches, by introducing an original constrained formulation of the image restoration inverse problem. Our formulation introduces an interpretable upper bound on the data fidelity term, making the restoration stage stable, reliable, and parameter-free. The problem resolution is then addressed with the local subspace Majorization-Minimization algorithm proposed in \cite{chouzenoux2022local}, for which we show the convergence in our context.
\end{itemize}
\end{enumerate}
We demonstrate the effectiveness of our two-step restoration pipeline on both simulated and real-world data acquired with an Olympus XLPLN25XWMP two-photon microscope. The associated Python code is publicly shared\footnote{\url{https://github.com/SegoleneMartin/biphoton}}, ensuring reproducible research.

The outline is as follows. Section \ref{sec:PSF} introduces our PSF estimation method, and validation experiments on synthetic data. In Section \ref{sec:resto}, we introduce our noise modeling and image restoration approaches, and we assess them on synthetic datasets. Lastly, in Section \ref{sec:mouse}, we present the experimental evaluation of our complete pipeline to the imaging of real mouse muscle data.

\section{A novel approach to PSF calibration using large beads}\label{sec:PSF}

In this section, we introduce our novel approach for estimating the PSF from acquisition of calibrated fluorescent beads with a diameter of $1\mu \text{m}$, larger than the microscope resolution. 

\subsection{Preparing the data}
For performing the calibration stage, multiple fluorescent beads with known size, are observed under the microscope in one comprehensive 3D image. The PSF estimation aims at fitting a model on each individual bead. To do so, the initial large image is cropped, to form several volumes of interest each containing an isolated bead. The cropping process we propose consists in applying a standard Wiener filtering to reduce noise, followed by a binarization using a predefined threshold (in practice, a given pourcentage of the maximum intensity), and an automatic connected component labeling algorithm \cite{wu2005optimizing}. Each selected component then defines a cropped region, to which we apply the PSF estimation method. Under the assumption of stationarity of the PSF model within the observed volume, the final PSF estimate can be obtained by averaging the models fitted in each volume interest, to get a more accurate overall estimation.
We now explicit our PSF estimation process, given a volume containing a single bead. An example of image of a polystyrene bead with $1$ $\mu m$ diameter (FluoSphere$^{\text{TM}}$ Carboxylate-Modified Microsphere, manufactured by Thermofisher, product number: F8823), acquired by our Olympus XLPLN25XWMP two-photon microscope is displayed in Figure \ref{fig:bead_image}.

\begin{figure}
    \centering
    \includegraphics[scale=0.4]{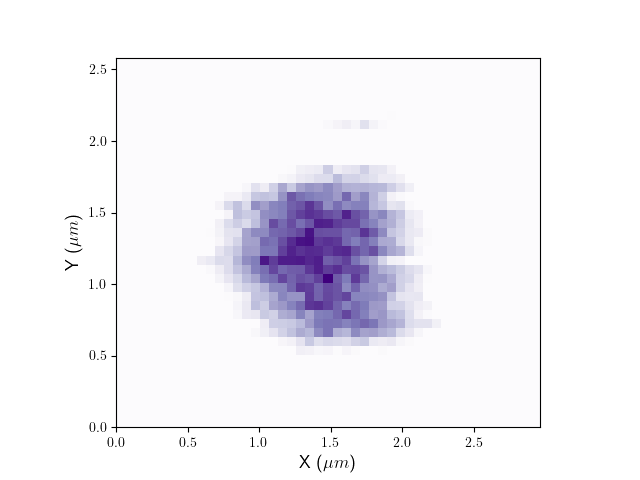} \includegraphics[scale=0.4]{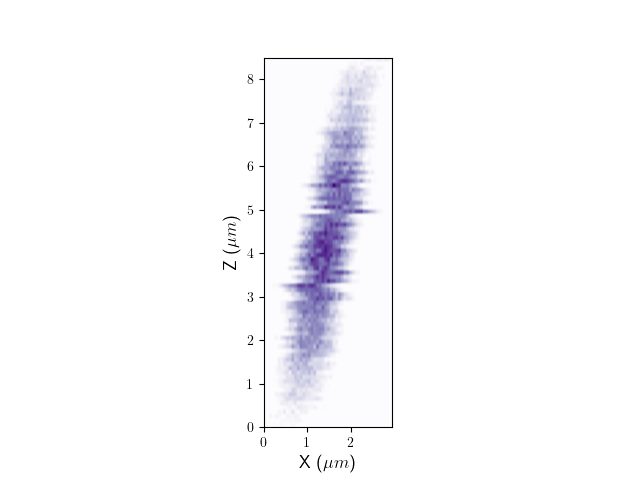}
    \caption{Two-photon acquisition of a polystyrene bead of $1\mu$m of diameter. XY plane section (left) and XZ plane section (right). The acquisition was carried out with an excitation wavelength of $810 \text{nm}$. The voxel size was fixed to $0.049 \times 0.049\times 0.1 \mu$m$^3$.}
    \label{fig:bead_image}
\end{figure}

\subsection{Modeling of the problem}

Since we are in a calibration context, the true bead characteristics are known. In particular, the experimenter selects the bead diameter, and a spherical shape can be assumed. Moreover, the position of the bead can be easily deduced by a basic identification of the centroid of the bead image. This allows us to build a complete model for the observed object, up to a potential multiplicative factor in the pixel intensities.

For the remainder of the paper, let us introduce the following notations: $\Vert \cdot \Vert$ represents the Euclidean norm on $\R^N$, $\mathbf{1}_N$ the unit vector of $\R^N$, $\Vert \cdot \Vert_{\rm F}$ the Frobenius norm on $\R^{N\times N}$, $2^{\mathcal{C}}$ the set of all subsets of a given set $\mathcal{C}$, $\mathcal{S}_3$ the set of symmetric matrices in $\R^{3\times 3}$, $\mathcal{S}_3^+$ the set of symmetric positive semi-definite matrices in $\R^{3\times 3}$, $\mathcal{S}_3^{++}$ the set of symmetric positive definite matrices in $\R^{3\times 3}$, and $\bs{I}_3$ the identity matrix of $\R^3$.

\paragraph{Continuous Modeling}

Let $x \colon \R^3 \longrightarrow \R$ represent a sphere with diameter $\tau$, modeling the bead fluorescence. The bead is assumed to emit with a uniform intensity, so we set, for all $\mathbf{u} \in \R^3$, $x(\bs{u}) = 1$ if $\Vert \mathbf{u} \Vert \leq \tau$, and $x(\bs{u}) = 0$ otherwise. The ground truth florescence intensity is then scaled by a factor $\bar{\beta}_{\mathrm{c}} \in (0, +\infty)$, which is usually unknown.

Let $y\colon \R^3 \longrightarrow \R$ be the image of the bead obtained with the MPM device. The observation model takes the form of \eqref{eq:blind_inverse_pb}. For the sake of simplicity, in this calibration step, we adopt an assumption of an additive i.i.d. Gaussian noise. Thus, we have:
\begin{equation}\label{eq:continuous_model}
(\forall \bs{u} \in \R^3) \quad  y(\bs{u}) = \bar{\alpha} + \bar{\beta}_{\mathrm{c}}(\bar{h} \star x)(\bs{u}) + \nu(\bs{u}),
\end{equation}
where $\nu$ is a function accounting for the Gaussian noise, $\bar{\alpha} \in \R$ is a background term, and $\bar{h} \in \mathrm{L}^1(\R^3)$ is a convolution kernel (i.e., the PSF) such that $\bar{h} (\bs{u}) \geq 0$ almost everywhere on $\R^3$ and satisfying
\begin{equation}
\int_{\R^3} \bar{h} (\bs{u}) \,\mathrm{d}\bs{u} =1.
\end{equation}
Given the observation model \eqref{eq:continuous_model}, the goal is to estimate the unknowns $\bar{\alpha}$, $\bar{\beta}_{\mathrm{c}}$, and $\bar{h}$. 

\paragraph{Discrete Modeling}
MPM acquisitions are done point-by-point, on a voxel grid with a given resolution. Hence, one only has access to a sampling of $y$ on a bounded discrete set $\Omega$ of $\R^3$, paved into $N \in \N$ voxels with mass centers $(\bs{\omega}_n)_{1 \leq n \leq N} \in (\R^3)^N$. The discretized version of the continuous problem \eqref{eq:continuous_model} then becomes:
\begin{equation}\label{eq:discrete_model}
 \bs{y} = \bar{\alpha}\mathbf{1}_N + \bar{\beta}(\bar{\bs{h}} \ast \bs{x})+ \bs{\nu},
\end{equation}
where $\bs{y}= (y(\bs{\omega}_n))_{1 \leq n \leq N} \in \R^N$ (resp. $\bs{x} =(x(\bs{\omega}_n))_{1 \leq n \leq N}\in \R^N$ and $\bs{\nu}=(\nu(\bs{\omega}_n))_{1 \leq n \leq N} \in \R^N$) is the discretization of $y$ (resp. $x$ and $\nu$) on $\Omega$, $ \bar{\bs{h}}=(\bar{h}_n)_{1\leq n \leq N} \in \R^N$ is a discrete convolution kernel satisfying $\bar{\bs{h}} \in \Delta_N$ where $\Delta_N$ is the simplex of $\R^N$ defined as
\begin{equation}\label{eq:def_simplex}
\Delta_N = \left\{\bs{h}=(h_n)_{1\leq n \leq N} \in \R^N\, \big| \, (\forall n \in \{1, \dots, N\})~h_n \geq 0 ~\text{and}~\sum_{n=1}^N h_n = 1\right\},
\end{equation}
and $\ast$ states for the discrete 3D convolution operator with appropriate padding (typically, circular padding). In addition, $\bar{\beta}>0$ is a scaling factor for the discrete model, which differs from the one in the continuous model, $\bar{\beta}_{\mathrm{c}}$. Indeed, observe that $\bar{\bs{h}}$ is not, per se, a discretization of $h$. In particular, an implicit rescaling is assumed to ensure the sum of $\bar{\bs{h}}$ equals $1$, for simplicity. 
The goal is thus to provide estimates of $\bar{\bs{h}}$, and the shift/scaling factors $(\bar{\alpha},\bar{\beta})$, given the knowledge of the measurements $\bs{y}$ and the bead model $\bs{x}$.

\subsection{Proposed minimization problem}

Following~\cite{chouzenoux2019optimal}, we introduce a prior so as to promote a PSF with a Gaussian shape. This leads to solving the following regularized variational problem:
\begin{equation}\label{eq:minimization_pb_PSF}
\underset{\genfrac{}{}{0pt}{}{\alpha \in \mathcal{A}, \, \beta \in \mathcal{B},}{ \bs{h} \in \Delta_N, \, \bs{D} \in \mathcal{S}_3^+}}{\mathrm{minimize}} ~\frac{1}{2} \Vert \bs{y} - \alpha \mathbf{1}_N -  \beta( \bs{h} \ast \bs{x}) \Vert^2 + \lambda \,
\mathcal{KL}\left(\bs{h} \,\Vert \, \zeta\bs{g}(\bs{D} + \epsilon_1 \bs{I}_3)\right) + \epsilon_2 \Vert \bs{D} \Vert_{\rm F}^2.
\end{equation}
Hereabove, $\epsilon_1 > 0$, so for every $\bs{D} \in \mathcal{S}_3^+$, matrix $\bs{D} + \epsilon_1 \bs{I}_3 \in \mathcal{S}_3^{++}$, and $\bs{g}(\bs{D} + \epsilon_1 \bs{I}_3) \in \R^N$ states for the discretization on $\Omega$ of the centered Gaussian density function with inverse covariance $\bs{D} + \epsilon_1 \bs{I}_3$, corresponding to
\begin{equation}\label{eq:def_gaussienne}
(\forall \bs{S} \in \mathcal{S}_3^{++})\,(\forall n \in \{1, \dots, N\}) \quad [\bs{g}(\bs{S})]_n = \sqrt{\frac{ |\bs{S}|}{(2\pi)^3}} \exp\left(-\frac{1}{2} \bs{\omega}_n^\top \bs{S}  \bs{\omega}_n \right).
\end{equation}
Moreover, $\lambda >0$ is a regularization parameter, and $\zeta\in (0, +\infty)$ is a known scaling parameter accounting for the measure of the discrete Gaussian function $\bs{g}(\bs{D} + \epsilon_1 \bs{I}_3)$ on the grid $\Omega$. When the grid is sufficiently fine, which we will assume throughout this discussion, $\zeta$ can be approximated by
\begin{equation}
\zeta \simeq r_X r_Y r_Z,
\end{equation} 
where $r = (r_X, r_Y, r_Z) \in (0, +\infty)^{3}$ represents the dimensions, in micrometers, of the voxels in the captured image (i.e., the image resolution). 
Sets $\mathcal{A} = [\alpha_{-}, \alpha_{+}]$ and $\mathcal{B}=[\beta_{-}, \beta_{+}]$ are real intervals corresponding to known bounds on $\alpha$ and $\beta$, respectively.
 $\mathcal{KL}$ denotes the Kullback-Leibler divergence between two discrete probability distributions, defined as
\begin{align}\label{eq:def_KL_div}
 (\forall (\bs{p}, \bs{q}) \in (\Delta_N \cap (0, +\infty))^2)\quad   \mathcal{KL} \left( \bs{p} \middle\| \bs{q} \right) &= \sum_{n=1}^{N}  p_n \mathrm{log} \left( \frac{p_n}{q_n }\right).
 \end{align}
 Finally $\epsilon_2 \in (0, +\infty)$ is an arbitrarily small penalty parameter.

 Let us denote, for conciseness:\begin{equation}
   (\forall \bs{h} \in \Delta_N) \, (\forall \bs{D} \in \mathcal{S}_3^+)\quad   \Psi(\bs{h}, \bs{D}) = \mathcal{KL}\left(\bs{h} \,\Vert \, \zeta\bs{g}(\bs{D} + \epsilon_1 \bs{I}_3)\right).
 \end{equation}
 To avoid potential numerical issues with the logarithmic term in $\mathcal{KL}$ divergence, we extend $\Psi$ into a twice continuously differentiable function $\tilde{\Psi}$ on $\Delta_N \times \mathcal{S}_3$, defined as
 \begin{multline}
          (\forall \bs{h} \in \Delta_N) \, (\forall \bs{D} \in \mathcal{S}_3)\quad \tilde{\Psi}(\bs{h}, \bs{D}) = \sum_{n=1}^N \Big( E(h_n) - h_n \ln \zeta  \\+ \frac{h_n}{2} \left(3 \ln(2 \pi) + \Phi(\bs{D}) + \bs{\omega}_n^\top (\bs{D}+ \epsilon_1 \bs{I}_3) \bs{\omega}_n \right)\Big),
 \end{multline}
 where, for any $\bs{D} \in \mathcal{S}_3$ diagonalized as $\bs{D} = \bs{U} \mathrm{Diag}(\bs{s})\bs{U}^\top$ with $\bs{U}$ an orthogonal matrix of $\R^{3 \times 3}$ and $\bs{s} \in \R^3$,
 \begin{align}\label{eq:def_Phi}
     \Phi(\bs{D}) =
       \sum_{i=1}^3  \phi(s_i),
 \end{align}
with $\phi$ defined as
 \begin{equation}
     (\forall t \in \R)\quad  \phi(t) = \syst{\begin{array}{ll}
       -  \ln(t + \epsilon_1) & \text{if } t \geq 0, \\
         - \ln(\epsilon_1) - \epsilon_1^{-1} t + \epsilon_1^{-2} t^2 & \text{otherwise,}
     \end{array}}
 \end{equation}
 and 
 \begin{align}
     (\forall u \in \R)\quad E(u) = \syst{\begin{array}{ll}
        u \ln (u)  & \text{if } u >0,\\
        0 & \text{if } u = 0,\\
        +\infty & \text{otherwise.}
     \end{array}}
 \end{align}
In order to model the constraints in Problem \eqref{eq:minimization_pb_PSF}, we introduce, for any non-empty, closed, convex set $C \subseteq \R^n$, 
its indicator function $\iota_C$, namely $\iota_{C}(x) = 0$ if $x\in C$, $\iota_{C}(x) = +\infty$ otherwise. In a nutshell, Problem \eqref{eq:minimization_pb_PSF} is equivalent to minimize, for $(\alpha,\beta,\bs{h},\bs{D}) \in \mathbb{R} \times \mathbb{R} \times \mathbb{R}^N \times \mathcal{S}_3$,
\begin{multline}\label{eq:cost_function}
F(\alpha, \beta, \bs{h}, \bs{D}) = \frac{1}{2} \Vert \bs{y} - \alpha \mathbf{1}_N -  \beta( \bs{h} \ast \bs{x}) \Vert^2 + \lambda \tilde{\Psi}(\bs{h}, \bs{D}) \\ + \iota_{\mathcal{A}}(\alpha)  + \iota_{\mathcal{B}}(\beta)  + \iota_{\Delta_N}(\bs{h})  +  \iota_{\mathcal{S}_3^+}(\bs{D}+ \epsilon_1 \bs{I}_3) + \epsilon_2  \Vert \bs{D} \Vert_F^2.
\end{multline}
It will be established in Proposition \ref{prop:existence_min} that, since $\epsilon_2 \in (0, +\infty)$,
$F$ admits a minimizer.

\subsection{Minimization algorithm}

The objective function in \eqref{eq:cost_function} is nonconvex. However, it is convex with respect to each variable. A common strategy for such scenarios consists of using alternating minimization techniques. At each iteration, we minimize 
$F$ with respect to a variable while the other variables are held constant. This simple method, also known as Block Coordinate Descent, has previously been utilized for PSF model fitting in microscopy, for instance in \cite{fortun2018reconstruction,li2018accurate}. Yet, closed-form updates are not always available. Moreover, convergence to an optimal solution using this method is only guaranteed when the partial functions on each variable and at each iteration are uniformly strongly convex \cite{luo1993error,beck2013convergence}. This stringent assumption is not met here. To ensure both efficiency and guarantees of convergence, a strategy enriched with proximal tools is preferable, as highlighted for instance in \cite{chouzenoux2019optimal,chouzenoux2016block,Bonettini2018}. The particular structure of \eqref{eq:cost_function} suggests a hybrid proximal alternating scheme (\cite{xu2013block,phan2023inertial,chouzenoux:hal-03591742}) that we present hereafter.

\subsubsection{Preliminaries}

In this section, we give some mathematical definitions that will be useful in the subsequent parts of the paper. Let $\mathcal{H}$ be a Hilbert space endowed with the scalar product $\langle \cdot, \cdot \rangle$.

\begin{definition}[Domain]
\cite[Def.1.4]{bauschke2019convex}
Let $f\colon \mathcal{H} \longrightarrow (-\infty, +\infty]$. The domain of $f$ is the set of all points in $\mathcal{H}$ where $f$ attains finite values. It is defined as
$\mathrm{dom } f = \{ x \in \mathcal{H} \mid f(x) < +\infty \}$.
A function $f$ is said to be proper if its domain is nonempty. The set of proper, lower semi-continuous, convex functions of $\mathcal{H}$ is denoted by $\Gamma_0(\mathcal{H})$.
\end{definition}

Next, we introduce the concept of the proximity operator. This tool will be at the core of our iterative algorithm.

\begin{definition}[Proximity operator]\cite[Def.12.23]{bauschke2019convex}
Let $f \in \Gamma_0(\mathcal{H})$. The proximity operator of $f$ at $x \in \mathcal{H}$ is defined as
$
\mathrm{prox}_f(x) = \underset{u \in \mathcal{H}}{\mathrm{argmin}} f(u) + \frac{1}{2} \| u-x \|^2$.
In particular, for a non-empty closed convex set $C$, the proximity operator of $\iota_C$ coincides with the orthogonal projection onto $C$.
\end{definition}

One of the foundational concepts of convex analysis is the subdifferential, which generalizes the concept of a derivative.

\begin{definition}[Moreau subdifferential]\cite[Def.16.1]{bauschke2019convex}
Let $f \in \Gamma_0(\mathcal{H})$. The Moreau subdifferential of $f$, denoted by $\partial f$, is 
\begin{align}
\partial f : \mathcal{H} &\longrightarrow 2^{\mathcal{H}} \nonumber\\
x &  \longmapsto \left\{ u \in \mathcal{H} \mid (\forall y \in \mathcal{H}) \langle y-x | u \rangle + f(x) \leq f(y) \right\}.
\end{align}
\end{definition}
When $f$ is both convex and Fr\'echet differentiable, its subdifferential at any point $x$ just contains its gradient.

Finally, we present some useful properties of the subdifferential and the proximity operator.

\begin{proposition}[]\cite[Chap.16]{bauschke2019convex}\label{prop:sub_prox}
Let $f\in \Gamma_0(\mathcal{H})$ and $g\in \Gamma_0(\mathcal{H})$.
\begin{enumerate}[label=(\roman*)]
    \item \label{prop:sub_prox_i}
If $\mathrm{int} (\mathrm{dom}\, g) \cap \mathrm{dom}\, f \neq \varnothing$ or $\mathrm{dom}\, g \cap \mathrm{int} (\mathrm{dom}\, f) \neq \varnothing$, then \begin{equation}
\partial f + \partial g = \partial(f + g).
\end{equation}
\item \label{prop:sub_prox_ii}
The following equivalence holds:
\begin{equation}
    (\forall x \in \mathcal{H}) \quad p = \prox_f(x) \, \iff \, x - p \in \partial f(p).
\end{equation}
\end{enumerate}
\end{proposition}

With these mathematical concepts explained, we are now ready to get into more details of our algorithm.

    \subsubsection{Proposed algorithm}
Guided by the methodology in \cite{xu2013block} and \cite{phan2023inertial}, we opt for an alternating minimization algorithm acting on the variables $\alpha$, $\beta$, $\bs{h}$, and $\bs{D}$. At each iteration of the algorithm, we update a specific variable employing one of the three following distinct schemes applied to the partial function with respect to this variable: an exact update step, a proximal step, or a proximal linearized step (also called forward-backward, or proximal gradient).




Within our framework, the variables $\alpha$ and $\beta$ undergo exact step updates, $\bs{D}$ is updated through a proximal point step, and $\bs{h}$ is updated through a proximal linearized step. We provide the expressions for these updates in the section that follows. Our algorithm, called GENTLE (\textbf{G}aussian k\textbf{E}r\textbf{N}el fi\textbf{T}ting using \textbf{L}arge B\textbf{E}ads), is summarized in Algorithm \ref{algo:MPM_PSF}, where we use the short notation
\begin{equation}\label{eq:def_F_tilde}
(\forall (\alpha,\beta,\bs{h},\bs{D}) \in \mathbb{R} \times \mathbb{R} \times \mathbb{R}^N \times \mathcal{S}_3) \quad \tilde{F}(\alpha, \beta, \bs{h}, \bs{D}) = \lambda \tilde{\Psi}(\bs{h}, \bs{D}) + \iota_{\Delta_N}(\bs{h}),
\end{equation}
and the stepsizes $(\gamma_h,\gamma_D) \in (0,+\infty)^2$. 
{In addition, we denote $\bs{X} \in \R^{N \times N}$ the block Toepliz matrix such that, for all $(\bs{u},\bs{x}) \in (\R^N)^2$, $ \bs{X} \bs{u} = \bs{x} \ast \bs{u}$.}

\vspace{1cm}
\begin{algorithm*}[H]\label{algo:MPM_PSF}
\caption{\textbf{GENTLE}: \textbf{G}aussian k\textbf{E}r\textbf{N}el fi\textbf{T}ting using \textbf{L}arge B\textbf{E}ads}
\textbf{Inputs:} Let $(\alpha^{(0)}, \beta^{(0)}, \bs{h}^{(0)}, \bs{D}^{(0)}) \in \mathcal{A} \times \mathcal{B}  \times \Delta_N \times \mathcal{S}_3^+$, $ (\gamma_h,\gamma_D) \in (0, +\infty)^2$. \\
\For{$\ell=0, 1, \dots$}{
$\alpha^{(\ell+1)} = \underset{\alpha \in \R}{\mathrm{argmin}} \, F(\alpha, \beta^{(\ell)}, \bs{h}^{(\ell)}, \bs{D}^{(\ell)}) $\\
$\beta^{(\ell+1)} = \underset{\beta \in \R}{\mathrm{argmin}} \, F(\alpha^{(\ell+1)}, \beta, \bs{h}^{(\ell)}, \bs{D}^{(\ell)}) $\\
$\bs{h}^{(\ell +1)} = \prox_{\gamma_{h} \tilde{F}}\left(\bs{h}^{(\ell)} - \gamma_{h} \beta^{(\ell+1)} {\bs{X}^\top} ( \bs{y} - \alpha^{(\ell+1)} \mathbf{1}_N -  \beta^{(\ell+1)} \bs{X} \bs{h}^{(\ell)})  \right)$\\
 $\bs{D}^{(\ell+1)} = \prox_{\gamma_{D}F(\alpha^{(\ell +1)}, \beta^{(\ell +1)},\bs{h}^{(\ell +1)}, \cdot)}(\bs{D}^{(\ell)})$}
 \textbf{Return:} $\alpha$, $\beta$, $\bs{h}$, $\bs{D}$.
\end{algorithm*}

    \subsubsection{Expressions of the updates}

    We now explicit the expressions for the four updates involved at each iteration of GENTLE. 

\begin{proposition}[Update on \texorpdfstring{$\alpha$}{a}]
Let $(\beta, \bs{h}, \bs{D}) \in  \R\times \R^N \times \mathcal{S}_3$. Then the minimizer of $F(\cdot,  \beta, \bs{h}, \bs{D}) $ at $\alpha$ is given by
\begin{equation}
\alpha =  \operatorname{proj}_{\mathcal{A}}\left( \frac{ 1}{N}( \bs{y} - \beta (\bs{h} \ast \bs{x}))^\top \mathbf{1}_N \right),
\end{equation}
where $\operatorname{proj}_{\mathcal{A}} $ is the orthogonal projection on the closed convex set $\mathcal{A}$, {reading $\operatorname{proj}_{\mathcal{A}} (\alpha) = \max\left(\alpha_{-}, \min\left(\alpha, \alpha_{+}\right)\right)$.}
\end{proposition}

\begin{proof}
Let $f= F(\cdot,  \beta, \bs{h}, \bs{D})$. Then, for any $\alpha \in \R$, 
\begin{align}
0 \in \partial f (\alpha) ~& \iff ~(\bs{y} - \beta (\bs{h}  \ast \bs{x}))^\top \mathbf{1}_N - \alpha N \in \partial \iota_{\mathcal{A}}(\alpha) \notag \\
& \iff ~ \frac{1}{N}(\bs{y} - \beta (\bs{h}  \ast  \bs{x}))^\top \mathbf{1}_N - \alpha  \in \partial \iota_{\mathcal{A}}(\alpha) \notag \\
& \iff ~ \alpha = \operatorname{proj}\left( \frac{ 1}{N}( \bs{y} - \beta (\bs{h} \ast \bs{x}))^\top \mathbf{1}_N \right),
 \end{align}
 where we used Proposition \ref{prop:sub_prox}(i) and (ii)
\end{proof}

\begin{proposition}[Update on \texorpdfstring{$\beta$}{b}]
Let $(\alpha, \bs{h}, \bs{D}) \in  \R\times \R^N \times \mathcal{S}_3$, and $\gamma_\beta \in (0,+\infty)$. Then the minimizer of of $ F(\alpha,  \cdot, \bs{h}, \bs{D}) $ is given by
\begin{equation}
 \beta = \operatorname{proj}_{\mathcal{B}} \left( \frac{( \bs{y} -\alpha  \mathbf{1}_N)^\top (\bs{h} \ast \bs{x}) }{\Vert \bs{h} \ast \bs{x}\Vert^2} \right).
\end{equation}
where $\operatorname{proj}_{\mathcal{B}} $ is the orthogonal projection on the closed convex set $\mathcal{B}$, {reading $\operatorname{proj}_{\mathcal{B}} (\beta) = \max\left(\beta_{-}, \min\left(\beta, \beta_{+}\right)\right)$.}
\end{proposition}
We skip the proof which is similar to the previous one.


\begin{proposition}[Update on \texorpdfstring{$\bs{h}$}{h}]\label{prop:prox_h}
Let $(\alpha, \beta, \bs{h}', \bs{D}) \in  \R^2 \times \R^N \times \mathcal{S}_3$, and $\gamma_h \in (0,+\infty)$.
Then the proximity operator of $\gamma_h \tilde{F}(\alpha,  \beta, \cdot, \bs{D})$ at $\bs{h}'$ is given by
\begin{equation}
\prox_{\gamma_h \tilde{F}(\alpha,  \beta,\cdot, \bs{D})} (\bs{h}') = (\rho^{-1} W\left(\rho \exp(w_n(\widehat{\mu}))\right)_{1\leq n \leq N},
\end{equation}
where $W$ denotes the Lambert-W function \cite{corless1996lambert}, 
\begin{equation}\label{e:defrho}
\rho = \frac{1}{\lambda \gamma_h},
\end{equation}
and, for every $n \in \{1, \dots, N\}$, $w_n$ is the function defined as
\begin{equation}
(\forall \mu \in \R)\quad w_n(\mu) = -1 - c_n + \rho (h_n' - \mu) + \ln\zeta,
\end{equation}
with 
\begin{equation}
c_n = \frac{1}{2} \left(3\log(2\pi)+ \Phi(D)+  \bs{\omega}_n^\top (\bs{D} + \epsilon_1 \bs{I}_3) \bs{\omega}_n \right).
\end{equation}
Moreover, $\widehat{\mu} \in \R$ is the unique zero of the function
\begin{equation}
(\forall \mu \in \R)\quad \kappa(\mu) = \rho^{-1} \sum_{n=1}^N W(\rho \exp(w_n(\mu))) -1.
\end{equation}
\end{proposition}

\begin{proof}
The function to minimize is $\tilde{F}(\alpha,  \beta, \zeta,\cdot, \bs{D})  + \frac{1}{2 \gamma_h} \Vert \cdot - \bs{h}' \Vert^2$. It is equivalent to minimizing the function $f$ defined by
\begin{align}
(\forall \bs{h} \in (0, +\infty)^N) \quad f(\bs{h}) &= \lambda \gamma_h \sum_{n=1}^N \left(  h_n \ln h_n - h_n \ln \zeta + c_n h_n  \right)
+  \iota_{\Delta_N}(\bs{h}) +\frac{1}{2} \Vert \bs{h} - \bs{h}' \Vert^2. 
\end{align}
The Lagrangian function associated with the minimization of $f$ reads
\begin{multline}
(\forall \bs{h} \in (0, +\infty)^N)(\forall \mu \in \R)\quad  \mathcal{L}(\bs{h}, \mu) = \lambda \gamma_h \sum_{n=1}^N \Bigg(h_n \ln h_n + h_n\left(-\ln \zeta + c_n  \right)\\ +  \frac{1}{2\lambda \gamma_h}(h_n - h_n')^2\Bigg) + \mu \left(\sum_{n=1}^N h_n -1\right).
\end{multline}
Since Slater's condition obviously holds, there exists $\hat{\mu} \in \R$ such that $(\bs{h}, \hat{\mu})$ is a saddle point of $\mathcal{L}$ \cite{bertsekas1995nonlinear}. By Fermat's rule \cite{bauschke2019convex}, $\bs{h}$ is thus obtained by finding a zero of the derivative of $\mathcal{L}(\cdot, \hat{\mu})$. This yields, for every $n \in \{1, \dots, N\}$,
\begin{align}
&\lambda \gamma_h \left(1 + \ln\hat{h}_n + c_n - \ln\zeta  \right)+ \hat{h}_n - h_n' + \hat{\mu} = 0,\notag \\
\iff ~&  \ln(\hat{h}_n)  + \rho \hat{h}_n = w_n(\hat{\mu}),\notag \\
\iff ~& \rho \hat{h}_n \exp(\rho \hat{h}_n) = \rho \exp(w_n(\hat{\mu})),
\end{align}
where $\rho$ is given by \eqref{e:defrho}
and $ w_n(\hat{\mu}) = -1 +\rho( h_n' -  \hat{\mu}) - c_n + \ln\zeta$. 

Finally, recalling that the $W$-Lambert function is such that $(\forall z \in \R)~W(z) \exp(W(z)) = z$, we deduce that
\begin{equation}
\hat{h}_n = \rho^{-1} W\left(\rho \exp(w_n(\hat{\mu}))\right).
\end{equation} 
In addition, one can obtain $\hat{\mu}$ 
from the linear equality constraint.
This amounts to finding a zero of the function $\kappa$ defined as
\begin{equation}
(\forall \mu \in \R) \quad \kappa(\mu)= \rho^{-1} \sum_{n=1}^N W\left(\rho \exp(w_n(\mu)\right) - 1.
\end{equation}
It was shown in \cite{chouzenoux2019optimal}, relying on the properties of the $W$-Lambert function, that $\kappa$ admits a unique zero.
\end{proof}
The Lambert-W function appearing in Proposition \ref{prop:prox_h} is commonly in the expression of the proximity operators of entropic functions \cite{lapin2017analysis,cherni2016proximity}. While its evaluation, requiring the solution of a transcendental equation, can be efficiently achieved using a Newton-based method, the composition of $W$ with the exponential can lead to arithmetic overflow for large inputs. To address this, we employ the asymptotic expansion $W(\exp(u)) \approx u - \log(u)$ for $u > 10^2$ \cite{corless1996lambert}.

\begin{proposition}[Update on \texorpdfstring{$\bs{D}$}{D}]
Let $(\alpha, \beta, \bs{h}, \bs{D}) \in  \R^2 \times (0, +\infty)^N \times \mathcal{S}_3$, and $\gamma_D \in (0,+\infty)$. Then the proximity operator of $\gamma_D F(\alpha,  \beta, \bs{h}, \cdot) $ at $\bs{D}$ is given by
\begin{equation}
\prox_{\gamma_D F(\alpha, \beta, \bs{h}, \cdot)}(\bs{D}) = \frac{1}{2} \bs{V} \mathrm{Diag}\left( \left(\max(\mu_i - \epsilon_1 + \sqrt{(\mu_i + \epsilon_1)^2 + 4m}, 0)\right)_{1\leq i \leq d}\right) \bs{V}^\top,
\end{equation}
where $\bs{\mu} = ( \mu_i)_{1 \leq i \leq 3}$ is a vector of eigenvalues of $(2 \epsilon_2 \gamma_D +1)^{-1}\bs{D}- \bs{S}$ and $\bs{V}$ is a $3\times 3 $ orthogonal matrix such that $(2 \epsilon_2 \gamma_D +1)^{-1}\bs{D}- \bs{S} = \bs{V} \mathrm{Diag}(\bs{\mu}) \bs{V}^\top$ with 
\begin{equation}
\bs{S} =  \frac{\gamma_D}{2(2 \epsilon_2 \gamma_D +1)} \lambda \sum_{n=1}^N h_n \bs{\omega}_n \bs{\omega}_n^\top,
\end{equation}
and $m = \frac{1}{2} \gamma_D \lambda$.
\end{proposition}
\begin{proof}
    Direct extension of~\cite[Prop.4]{chouzenoux2019optimal}.
\end{proof}

\subsubsection{Convergence analysis} 

We now establish the convergence of GENTLE. We first show that the considered minimization problem has at least one solution.
\begin{proposition}\label{prop:existence_min}
    The cost function $F$ defined is \eqref{eq:cost_function} is lower-bounded and admits a minimizer.
\end{proposition}

\begin{proof}
    It is clear that $F$ is lower semi-continuous on $\R \times \R \times \R^N \times \mathcal{S}_3$. Let us now show that $F$ is coercive.
    Since $\mathcal{A}$, $\mathcal{B}$, and $\Delta_N$ are bounded subsets, it suffices to show that, for any $(\alpha, \beta, \bs{h}) \in \mathcal{A} \times \mathcal{B} \times \Delta_N$, 
    \begin{equation}
        F(\alpha, \beta, \bs{h}, \bs{D}) \underset{\genfrac{}{}{0pt}{}{\Vert \bs{D} \Vert_{\rm F} \rightarrow +\infty}{ \bs{D} \in \mathcal{S}_3^+}}{\longrightarrow} +\infty.
    \end{equation}
The following lower bound holds for $F$:
\begin{align}\label{eq:lower_bound_F}
(\forall \bs{D} \in\mathcal{S}_3)\quad     F(\alpha, \beta, \bs{h}, \bs{D}) &\geq \lambda c_1 + \frac{\lambda}{2} \sum_{n=1}^N h_n \left(3 \ln(2\pi) +  \Phi(\bs{D}) + \bs{\omega_n}^\top (\bs{D} +\epsilon_1 \bs{I}_3) \bs{\omega_n} \right) \nonumber\\ 
&\qquad \qquad + \epsilon_2 \Vert \bs{D} \Vert_F^2\nonumber \\
    & \geq \lambda c_1 + \frac{3\lambda}{2} \ln(2\pi) + \frac{\lambda}{2} \Phi(\bs{D}) + \epsilon_2 \Vert \bs{D} \Vert_F^2, 
\end{align}
where $c_1 = \inf \left\{\sum_{n=1}^N E(h_n) - h_n \ln\zeta\mid \bs{h} \in \Delta_N \right\}$ and  we have used the fact that $\bs{D} +\epsilon_1 \bs{I}_3 \succ 0$. Finally, given the definition of $\Phi$ in \eqref{eq:def_Phi}, it is clear the lower bound in \eqref{eq:lower_bound_F} goes to $+\infty$ when $\Vert \bs{D} \Vert_{\rm F} \rightarrow +\infty$. Therefore, $F$ admits a minimizer.
\end{proof}

Let us now establish the convergence of the iterates generated by Algorithm \ref{algo:MPM_PSF}. We base our analysis on the convergence result stated in \cite[Thm. 2]{xu2013block}. 

\begin{theorem}
    Let $(\alpha^{(0)}, \beta^{(0)}, \bs{h}^{(0)}, \bs{D}^{(0)}) \in \mathcal{A} \times \mathcal{B} \times \Delta_N \times \mathcal{S}_3^{+}$ be an initial point. For every $\ell \in \N$, let $\bs{t}^{(\ell)} = (\alpha^{(\ell)}, \beta^{(\ell)}, \bs{h}^{(\ell)}, \bs{D}^{(\ell)})$ be the sequence generated by Algorithm \ref{algo:MPM_PSF}. Assume that $\gamma_h< \frac{2}{\overline{L}}$ with $\overline{L} = \beta_{+}^2 \max \left\{ (|\operatorname{DFT}(\bs{x})_n |^2)_{1 \leq n \leq N}\right\}$ and $\operatorname{DFT}$ denotes the discrete Fourier transform. Assume that the grid is discretized finely enough such that the true bead $\bs{x}$ has at least one pixel with unit intensity.
    Then $(\bs{t}^{(\ell)})_{\ell \in \N}$ converges to a critical point of the objective function \eqref{eq:cost_function}, $\hat{\bs{t}} = (\hat{\alpha}, \hat{\beta}, \hat{\bs{h}}, \hat{\bs{D}})$.
\end{theorem}

\begin{proof}
    We show that the assumptions required by \cite[Thm. 2]{xu2013block} are satisfied. 
\begin{itemize}
    \item \textbf{Splitting of the objective function}. We first split the objective function $F$ into 1) a coupling term $f\colon \R \times \R \times \R^N \times \mathcal{S}_3$ which is block convex, differentiable with a Lipschitz continuous gradient on bounded subsets, 2) a separable sum of functions $g_\alpha \in \Gamma_0(\R)$, $g_\beta \in \Gamma_0(\R)$, $g_h \in \Gamma_0(\R^N)$ and $g_D \in \Gamma_0(\mathcal{S}_3)$ such that:
    \begin{equation}
        F(\alpha, \beta, \bs{h}, \bs{D}) = f(\alpha, \beta, \bs{h}, \bs{D}) + g_\alpha(\alpha) + g_\beta(\beta) + g_h(\bs{h}) + g_D(\bs{D}).
    \end{equation}
To do so, we define 
\begin{multline}
    (\forall (\alpha, \beta, \bs{h}, \bs{D}) \in \R \times \R \times \R^N \times \mathcal{S}_3) \\
    f(\alpha, \beta, \bs{h}, \bs{D}) = \frac{1}{2} \Vert \bs{y} - \alpha \mathbf{1}_N -  \beta( \bs{h} \ast \bs{x}) \Vert^2 + \frac{\lambda}{2}  \sum_{n=1}^N h_n \left( \Phi(\bs{D}) + \bs{\omega}_n^\top (\bs{D} + \epsilon_1 \bs{I}_3) \bs{\omega}_n\right),
\end{multline}
\begin{equation}
    g_\alpha = \iota_\mathcal{A}, \quad g_\beta = \iota_\mathcal{B}, \quad g_D=\iota_{\mathcal{S}_3^+} + \epsilon_2 \Vert \bs{D} \Vert_{\rm F}^2,
\end{equation}
and
\begin{equation}
    (\forall \bs{h} \in \R^N)\quad g_h(\bs{h}) = \lambda \sum_{n=1}^N \left( E(h_n) - h_n \ln(\zeta) + \frac{3 h_n}{2} \ln(2 \pi)\right). 
\end{equation}

\item \textbf{Properties of the partial functions}. 
We first need to check that the partial functions with respect to $\alpha$ and $\beta$ are strongly convex along the iterates $(\bs{t}^{(\ell)})_{\ell \in \N}$, with modulus independent of $\ell$. It is clear that it is the case for the partial function with respect to variable $\alpha$. On the other hand, for any $\ell \in \N$, let us denote $F_\mathsf{B}^{(\ell)}$ the function $\beta \mapsto F(\alpha^{(\ell+1)}, \beta, \bs{h}^{(\ell)}, \bs{D}^{(\ell)})$. Then, since the bead $\bs{x}$ has at least one pixel value equal to 1, say at $\omega_i$, $i\in \{1, \dots, N\}$,
\begin{align}
    \nabla^2 F_\mathsf{B}^{(\ell)} (\beta) &= \Vert \bs{h}^{(\ell)} \ast \bs{x}\Vert^2 \nonumber\\
    &\geq \Vert \bs{h}^{(\ell)} \ast \bs{\delta}_i \Vert^2 \nonumber \\
    &= \Vert  \bs{h}^{(\ell)} \Vert^2, 
\end{align}
with $\bs{\delta}_i \in \R^N$ the image corresponding to one illuminated pixel with intensity equal to 1 at $\omega_i$, and 0 elsewhere. The inequalitiy holds because $\bs{h}^{(\ell)} \in (0,+\infty)^N$. Moreover, since $\Vert \cdot \Vert \geq \frac{1}{N}\Vert \cdot \Vert_1$ and $\bs{h}^{(\ell)}\in \Delta_N$, we deduce that $\nabla^2 F_\mathsf{B}^{(\ell)} (\beta) \geq \frac{1}{N^2}$.
Therefore, $F_\mathsf{B}^{(\ell)}$ is $\frac{1}{N^2}$-strongly convex. Secondly, for the proximal linearized step on variable $\bs{h}$, the step-size $\gamma_h$ has to satisfy the condition:
\begin{equation}
    \gamma_h < \frac{2}{L},
\end{equation}
where $L$ is the Lipschitz constant of the gradient of the differentiable component $\bs{h} \longmapsto \frac{1}{2} \Vert \bs{y} - \alpha \mathbf{1}_N -  \beta( \bs{h} \ast \bs{x}) \Vert^2$. It can easily been shown that $L \leq \overline{L}$. 

\item \textbf{Boundedness of the sequence.} Let us demonstrate that the sequence $(\bs{t}^{(\ell)})_{\ell \in \N}$ is bounded. 
Since GENTLE alternates exact, proximal and proximal linearized updates, the sequence $(F(\bs{t}^{(\ell)}))_{\ell \in \N}$ is non-increasing. In addition, the objective function is coercive as seen in Proposition \ref{prop:existence_min}. This implies that $(\bs{t}^{(\ell)})_{\ell \in \N}$ is bounded.

\item \textbf{Kurdyka-\L{}ojasiewicz's inequality.} Lastly, the objective function $F$ must satisfy the so-called Kurdyka-\L{}ojasiewicz's inequality \cite{kurdyka1998gradients,bolte2007lojasiewicz}. This property is satisfied by a wide range of functions in the context of image processing. In the present case, we omit the proof that $F$ satisfies the Kurdyka-\L{}ojasiewicz's inequality as a very similar proof to the one in \cite{chouzenoux2019optimal} could be driven.
\end{itemize}
 Finally, applying \cite[Thm. 2]{xu2013block}, we deduce that $(\bs{t}^{(\ell)})_{\ell \in \N}$ converges to a critical point of the objective function $F$.
{In practice, we assess the convergence using the stopping criterion $\Vert \bs{t}^{(\ell+1)} -\bs{t}^{(\ell)}\Vert \leq \varepsilon$ with $\varepsilon=10^{-7}$.}


\end{proof}

\subsection{Validation of the PSF calibration method}\label{sec:validation_PSF}

In this section, we present experiments conducted to validate the proposed PSF calibration method GENTLE. These experiments have been designed to test the performance of our method both in simulated scenarios and in real-life settings. In all our experiments, the regularization parameter $\lambda$ is determined through an empirical grid search. Specifically, we select the $\lambda$ that minimizes the criterion $\Vert \bs{y} - \hat{\alpha} - \hat{\beta} \bs{g}(\hat{\bs{D}} + \epsilon_1 \bs{I}_3) \ast \bs{x} \Vert^2$, given estimates $(\hat{\alpha},\hat{\beta},\hat{\bs{D}})$ produced by GENTLE for a certain $\lambda$. {The other hyper-parameters are set to $(\alpha_{-}, \alpha_{+}) = (0, 1)$, $(\beta_{-}, \beta_{+}) = (0, 3)$, and $\epsilon_1=\epsilon_2= 10^{-6}$.}

\subsubsection{Experiments in simulated scenarios}\label{sec:exp_simu_beads}

We simulate a 3D synthetic bead image $\bs{x} \in \R^N$ of $1 \mu \text{m}$ of diameter on a regularly spaced grid with size $N = 40 \times 40 \times 80$ and voxel dimension 
$0.05 \times 0.05 \times 0.1 \mu \text{m}^3$ (which is a typical resolution grid in MPM). The observation $\bs{y} \in \R^N$ is then simulated as $\bs{y} = \bar{\alpha} + \bar{\beta} \bar{\bs{h}} \ast \bs{x} + \bs{\nu}$, for fixed values of $ \bar{\alpha} \in [0, +\infty)$, $\bar{\beta} \in [0, +\infty)$, $\bar{\bs{h}} \in \Delta_N$ and $\bs{\nu} \in \mathbb{R}^N$ the realization of a zero-mean Gaussian noise, with standard deviation chosen so as to obtain a given input signal-to-noise ratio (SNR).

We choose generalized exponential distributions shapes for the ground truth PSF $\bar{\bs{h}}$. 
This family of distributions not only includes Gaussian ones, but also effectively captures those with shorter or longer tails. In microscopy, the PSF often deviates from a ideal Gaussian form, optical aberrations being a notable factor influencing this shape variation. \cite{stallinga2010accuracy}.
These distributions are uniquely defined by a matrix $\bar{\bs{S}} \in \mathcal{S}_3^+$ and a parameter $\eta>0$, as
\begin{equation}
    \bar{\boldsymbol{h}} \propto  \left(\exp\left\{-\frac{1}{2} \left(\boldsymbol{\omega}_n^\top (\bar{\bs{S}})  \boldsymbol{\omega}_n \right)^{\eta/2}\right\} \right)_{1\leq n \leq N},
\end{equation}
where $\bar{\bs{h}}$ is normalized so that $\bar{\bs{h}}\in \Delta_N$. When $\eta=2$, we retrieve a Gaussian distribution.
The matrix $\bar{\bs{S}}$ can be designed to represent a kernel with a specific inclination and width. Formally, $\bar{\bs{S}}$ can be decomposed in an orthonormal basis so that it only depends on the second and third Euler angles $(\bar{\theta}, \bar{\varphi}) \in [0, \pi]\times [-\pi, \pi]$ and on the eigenvalues $\bar{\bs{s}} \in (0, +\infty)^3$ for each principal direction. Indeed, for symmetry reasons, the first Euler angle, corresponding to a rotation around the vertical axis, has no effect so can be set to 0. We thus adopt the notation $\bar{\bs{S}} = \bs{S}(\bar{\theta}, \bar{\varphi},\bar{\bs{s}})$. GENTLE provides an estimate $\hat{\bs{D}}$ such that $\hat{\bs{D}} + \epsilon_1 \bs{I}_N$ approaches $\bs{S}(\bar{\theta}, \bar{\varphi},\bar{\bs{s}})$. On a desktop machine having an Intel i7-4790, 4 CPU cores and 16 GB of RAM, the PSF estimation for one bead crop of size $50 \times 50 \times 90$ takes approximately $1$ hour to reach the stopping criterion $\Vert \bs{h}^{(\ell+1)} -\bs{h}^{(\ell)}\Vert \leq 10^{-6}$.

We draw a comparison of our method GENTLE with another method, which will refer to as Nonlinear Least Squares (NLS) and which has been extensively used in MPM \cite{moraes2008improving,kirshner2013psf,zhu2013efficient}. It corresponds to considering directly the following non-regularized and non-multiconvex problem:
\begin{equation}
\underset{\genfrac{}{}{0pt}{}{\alpha \in \mathcal{A}, \, \beta \in \mathcal{B}, }{\theta \in [0, \pi] ,\, \varphi \in [-\pi, \pi], \,\bs{s} \in (0, +\infty)^3}}{\mathrm{minimize}} ~\frac{1}{2} \Vert \boldsymbol{y} - \alpha \mathbf{1}_N -  \beta  \boldsymbol{g}(\bs{S}(\theta, \varphi, \bs{s})) \ast \boldsymbol{x} \Vert^2, \end{equation}
where $\boldsymbol{g}(\bs{S}(\theta, \varphi,  \bs{s}))$ is the normalized discrete Gaussian kernel on $\Omega$ defined with the Euler angles $(\theta, \varphi) $ and the eigenvalues $\bs{s}$.
This problem can be tackled using a Levenberg-Marquardt solver \cite{marquardt1963algorithm}. Note that, despite its popularity, NLS has its shortcomings, notably it lacks convergence guarantees. 

Our results are reported in Figure \ref{fig:comparison_GENTLE_NLS} using the {Percent root-mean-square difference (PRD) metric, defined as $100\times \cdot\Vert \hat{\alpha} + \hat{b}(\hat{\bs{h}}\ast \bs{x}) - \bar{\alpha} - \bar{b}(\bar{\bs{h}}\ast \bs{x})  \Vert \big/ \Vert  \bar{\alpha} + \bar{b}(\bar{\bs{h}}\ast \bs{x}) \Vert$. In the case of GENTLE, variable $\hat{\bs{h}}$ is obtained as an output, while for NLS, we set $\hat{\bs{h}} = \boldsymbol{g}(\bs{S}(\hat{\theta},\hat{\varphi}, \hat{\bs{s}}))$.} {For this experiment, the true Euler angles were set to $(\bar{\theta}, \bar{\varphi}) = (5\pi/6, \pi/6)$, the true eigenvalues to $\bar{\bs{s}} = (138.6, 138.6, 3.2)$, and the voxels dimensions to $0.049\times 0.049 \times 0.1 \mu$m$^3$. The results were averaged over 10 random noise realisations.} 
As one can observe, NLS performance degrades significantly for values of the model exponent $\eta$ lower or larger than $2$, i.e., when there is a mismatch between the assumed Gaussian model and the true one. By contrast, although promoting Gaussian shapes using a regularization strategy,  GENTLE manages to provide accurate and robust estimations even in the presence of non-Gaussian generalized exponential PSF shapes. 


\begin{figure}[H]
\includegraphics[scale=0.5]{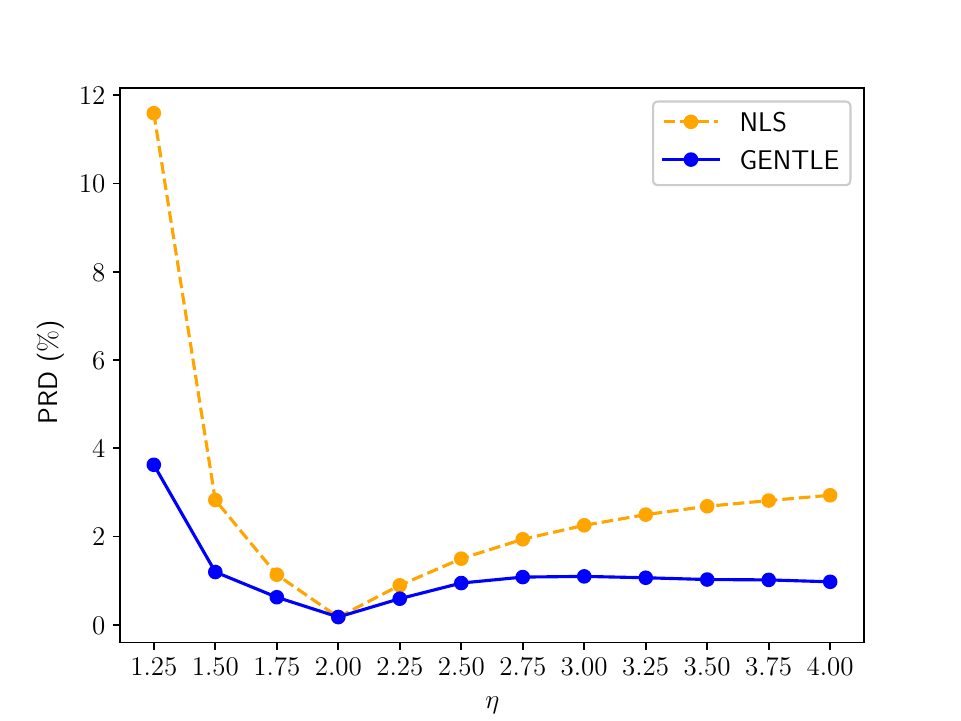}
\caption{Percent root-mean-square difference (PRD) as a function of the parameter $\eta$, for both GENTLE and NLS, for a noise level corresponding to SNR=10dB.\label{fig:comparison_GENTLE_NLS}}
\end{figure}

\subsubsection{Experiments on real beads in homogeneous medium}\label{subsubsec:exp_beads_opt}

The next step of our validation process consists in testing our method on real beads acquired under optimal conditions, namely, beads within in a homogenous medium. Such setting is typical in MPM device calibration, although it usually involves sub-resolution beads instead of $1\mu \text{m}$ diameter beads. The PSF depends on the system optical configuration, including the lens properties and acquisition settings \cite{diaspro2001confocal,cole2011measuring}.
We evaluate the PSF estimation using the Full Width at Half Maximum (FWHM) as a metric, as it is commonly done in the microscopy community. Measured in micrometers for our study, the FWHM represents the width of a bell-shaped curve when it is at half peak value.
 In the ideal case of multiphoton acquisitions where the optics are assumed ``perfect'', the values of the FWHM along each principal axis $X'$, $Y'$ and $Z'$ are~\cite{diaspro2001confocal}:
\begin{equation}\label{eq:FWHM_theo_XY}
    \text{FWHM}_{X'} \simeq \text{FWHM}_{Y'} \simeq \frac{0.7 \lambda_{\text{em}}}{\text{NA}},
\end{equation}
\begin{equation}\label{eq:FWHM_theo_Z}
    \text{FWHM}_{Z'} \simeq\frac{2.3 \lambda_{\text{em}} n_r}{\text{NA}^2},
\end{equation}
where $\lambda_{\text{em}}$ is the emission wavelength, $n_r$ is the refractive index of the immersion medium and $\text{NA}$ the numerical aperture. The above formula will serve as a reference to check the consistency of GENTLE results. 

For our experiment, we imaged a distilled water solution ($n_r = 1.33$) containing fluorescent polystyrene microspheres with a diameter of $\tau = 1\mu \text{m}$, marked with a yellow-green fluorophore emitting at 515 nm. The acquisition was performed 
with an excitation wavelengh $\lambda_{\text{exc}} = 810$ nm, a numerical aperture $\text{NA} = 1.05$, and a voxel size of $0.037\times 0.037 \times 0.05 \,\mu \text{m}^3$. 

In the obtained image, we were able to select four individual volumes of interest containing isolated beads, using our cropping procedure, on which we applied GENTLE with $\epsilon_1=\epsilon_2= 10^{-6}$, so yielding estimates $(\hat{\alpha},\hat{\beta},\bs{H},\hat{\bs{D}})$.
Then, performing a singular value decomposition of $\hat{\bs{D}} + \epsilon_1 \mathbf{I}_3$ and using the trigonometry formulas from \cite{depriester2018computing}, we deduce the Euler angles $(\hat{\theta}, \hat{\varphi}) \in [0, \pi]\times [-\pi, \pi]$ and the eigenvalues $\hat{\bs{s}}= (\hat{s}_{X'}, \hat{s}_{Y'}, \hat{s}_{Z'}) \in (0, +\infty)^3$ along each principal axis (where re-ordering is performed to align the axis with the Cartesian grid $(X,Y,Z)$)
The  FWHM along each axis $A \in \{X',Y',Z'\}$ is obtained as $\text{FWHM}_A =  2 \sqrt{2 \ln(2) / \hat{s}_A}$.

As Figure~\ref{fig:slice_bead_results} demonstrates, the estimated FWHM and angles are consistent across the four beads, validating the assumption of a stationary PSF in the observed field. Moreover, the average FWHM values are close to the theoretical expectations computed with \eqref{eq:FWHM_theo_XY} and \eqref{eq:FWHM_theo_Z}, equal to $0.34\mu\text{m}$ along the axial axes, and $1.43\mu \text{m}$ in the depth direction, which validates our approach.

\begin{figure}[H]
\centering
\includegraphics[width=0.45\linewidth]{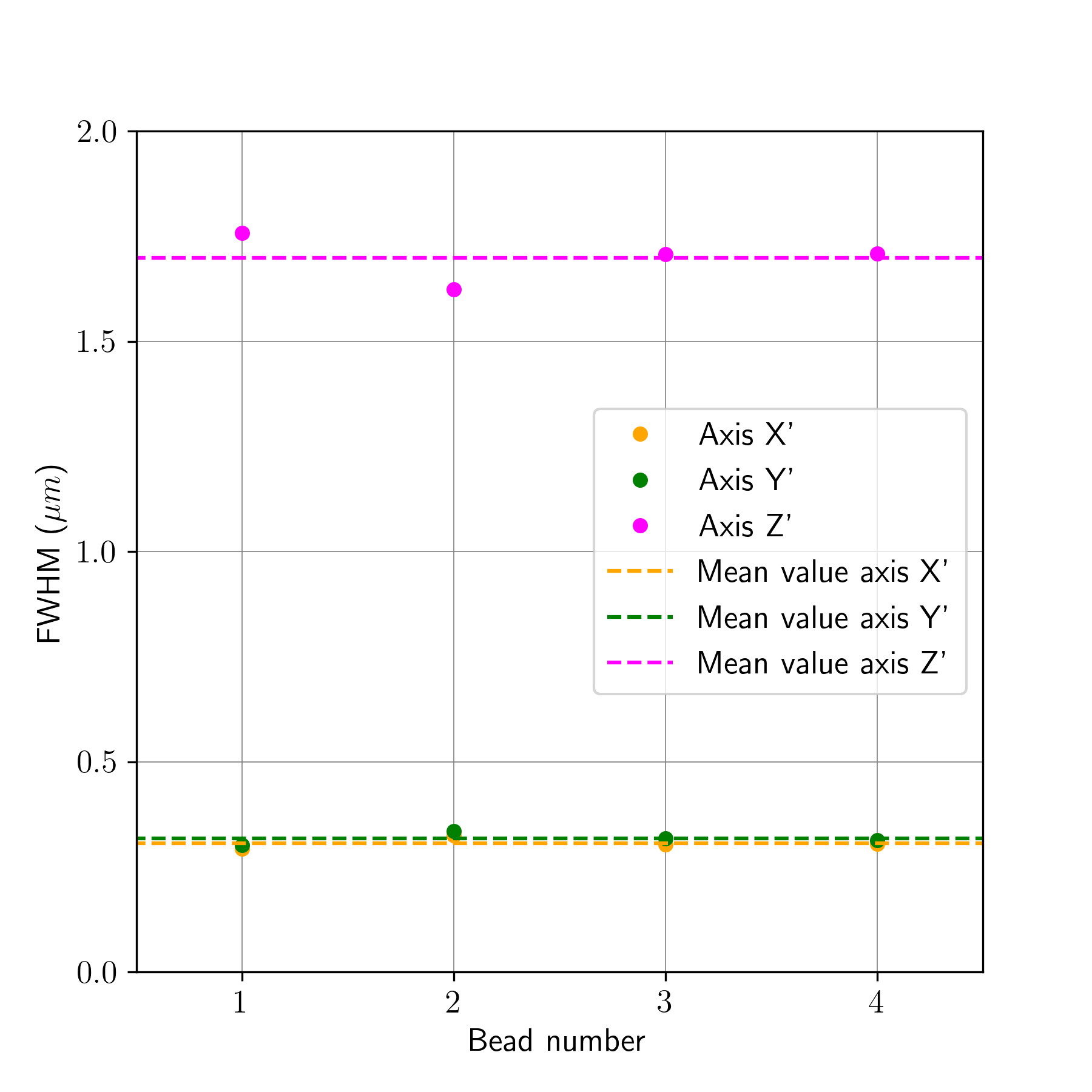} 
\includegraphics[width=0.45\linewidth]{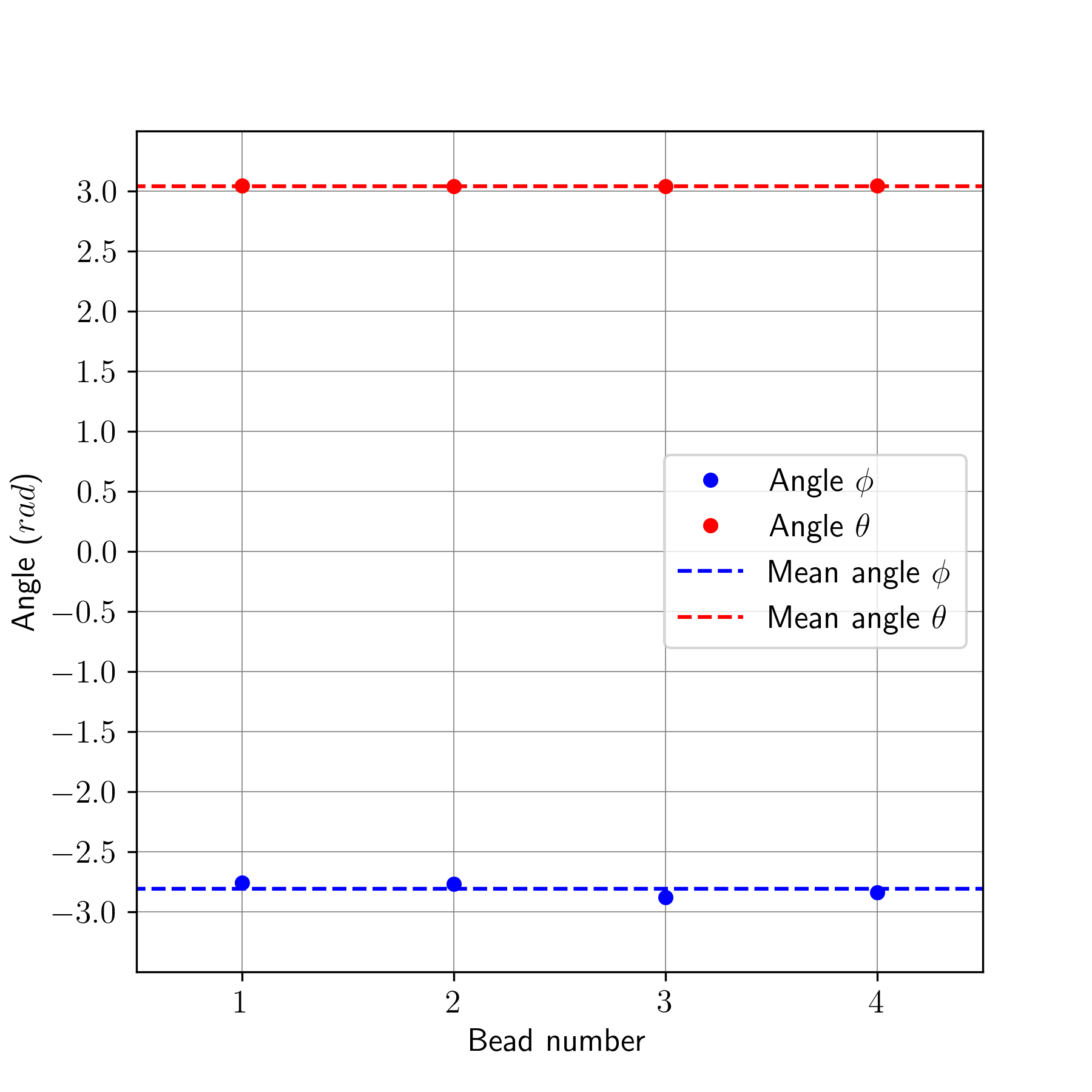}
\caption{Estimated FWHM (left) in the 3 principal component directions $\{X', Y',Z'\}$, and associated Euler angles (right), for the four isolated beads.}
\label{fig:slice_bead_results}
\end{figure}

\section{Proposed solution for the inverse problem}
\label{sec:resto}
We now shift our attention to the image restoration phase, where a non-calibrated and potentially complex, object, denoted by $\bar{\bs{x}} \in \mathbb{R}^M$, is observed, through the multi-photon microscope. The (degraded) observation is denoted $\bs{y} \in \mathbb{R}^M$. Note that the problem dimension $M$ usually differs from the dimension associated with the prior PSF estimation problem, typically with $M >N$, as the object of interest is typically much more spread than a micrometric bead. The goal is to accurately reconstruct an estimate of  $\bar{\bs{x}} \in \mathbb{R}^M$ given $\bs{y} \in \mathbb{R}^M$, and our knowledge of the PSF, deduced from the approach from Section \ref{sec:PSF}. Note that, in a practical scenario, the calibration phase could have been performed at another grid resolution, which would require an interpolation step, before using the estimated PSF from GENTLE algorithm. For the sake of simplicity, we opt here for a simple (though common) setting where the PSF $\bs{h} \in \Delta_N$ has been estimated under the same voxel resolution than the acquisition of the new object. GENTLE method also provides the background parameter, $\alpha \in \mathbb{R}$. The inverse problem thus reads,
\begin{equation}\label{eq:inverse_pb_deconvolution}
    y = \mathcal{D}(\bs{H} \bar{\bs{x}} + \alpha),
\end{equation}
where $\mathcal{D}\colon \mathbb{R}^M \rightarrow \mathbb{R}^M$ represents the noise model and $\bs{H} \in \R^{M\times M}$ is the linear operator such that, for every $\bs{u} \in \R^M$, $ \bs{H} \bs{u} = \bs{h} \ast \bs{u}$. Here, the convolution product is performed with zero-padding. While the calibration step is usually performed in ideal acquisition conditions, where the Gaussian additive noise assumption is sufficient, the situation differs in the image restoration stage (for example, \emph{in vivo} acquisitions might require low laser power, and as such, leads to low-photon counting). It then becomes necessary to introduce a more realistic noise model, to guarantee high quality restored images. 
In the following we introduce the considered noise model, we propose a direct method for estimating noise parameters, as well as a variational approach paired with an algorithmic solution to address the inverse problem \eqref{eq:inverse_pb_deconvolution}.

    \subsection{A heteroscedastic noise model for MPM}

Many restoration techniques for MPM imaging assume  a Gaussian homoscedastic noise (i.e., with constant variance across the whole image) \cite{lefort2021famous,danielyan2014denoising,doi2018high}. While such assumption simplifies computational needs, it might not always align with the reality of fluorescence microscopy imaging. The Gaussian assumption generally holds true for high-exposure scenes with abundant photons per the central limit theorem. In contrast, multiphoton or confocal imaging often exhibits a Poisson or Poisson-Gaussian noise profile \cite{crivaro2011multiphoton}. In \cite{van1997quantitative}, the authors illustrate the benefits of adopting a Poisson noise model over a Gaussian one in terms of restoration quality. However, handling such noise types can be computationally intensive, as it requires minimization of the non-smooth Poisson data-fidelity function \cite{zanella2009efficient,carlavan2011sparse,chouzenoux2015convex}. The mixed Poisson-Gaussian noise model describes the corrupted image $\bs{y}=(y_m)_{1\leq m \leq M}$ as
\begin{equation}\label{eq:poisson_gaussian}
(\forall m \in \{1, \dots, M\})\quad y_m \sim \alpha + a \mathcal{P}(a^{-1} [\bs{H} \bar{\bs{x}}]_m)+  \mathcal{N}(0, b),
\end{equation}
with some noise parameters $a \in (0, +\infty)$ and $b \in (0, +\infty)$. Such noise consists of the mixture of a multiplicative noise with mean equal to the intensity of the pixel and a variance proportional to the intensity of the pixel up to the fixed scale $a$, and an additive Gaussian noise identically distributed in each pixel, corresponding to a background noise with variance $b$. However, as shown in \cite{chouzenoux2015convex}, the resulting likelihood term is complicated, involving intractable series that need to be approximated, at the price of the computational time.

In the present work, we consider instead an additive heteroscedastic (i.e., with variance depending on the pixel) noise model approximating the Poisson-Gaussian noise, in the line of the method proposed in \cite{foi2008practical,foi2009clipped} for low-exposure natural images. To transform the Poisson-Gaussian noise in \eqref{eq:poisson_gaussian} into a fully additive one, we use the following well-known result: if $X$ is a random variable following a Poisson distribution $\mathcal{P}(\theta)$, with $\theta >0$, then 
\begin{equation}
\frac{1}{\sqrt{\theta}}(X - \theta) \underset{\theta \rightarrow +\infty}{\overset{\mathcal{L}}{\longrightarrow}} \bar{X},
\end{equation}
where $\bar{X} \sim \mathcal{N}(0, 1)$ and $\mathcal{L}$ denotes the convergence in distribution. In other words, for large values of $\theta$, the approximation $\mathcal{P}(\theta) \simeq \mathcal{N}(\theta, \theta)$ is valid. Therefore, assuming the multiplicative and background noises are independent, the final model of noise we consider is Gaussian, with a variance varying linearly with the intensity of the pixel, i.e.,
\begin{equation}
(\forall m \in \{1, \dots, M\})\quad y_m = \alpha + [\bs{H} \bar{\bs{x}}]_m  + w_m,
\end{equation}
where $(w_m)_{1\leq m \leq M}$ are independent variables sampled according to
\begin{equation}\label{eq:noise_model}
(\forall m \in \{1, \dots, M\})\quad w_m \sim \mathcal{N}(0, \sigma^2([\bs{H} \bar{\bs{x}}]_m)),
\end{equation}
and $\sigma$ is the function corresponding to the standard deviation, defined as
\begin{equation}\label{eq:linear_var}
(\forall t \in \R)\quad \sigma(t) = \syst{ \begin{array}{l l} \sqrt{a t + b}  & \text{if } t\geq 0,  \\ 0& \text{otherwise. } \end{array}}
\end{equation}
The noise model \eqref{eq:noise_model} depends on two parameters $a$ and $b$, for which we present hereafter an estimation strategy.

    \subsection{Estimation of the noise parameters}

We propose in Algorithm \ref{algo:noise_params} our procedure for estimating heteroscedastic Gaussian noise parameters in 3D images. {Our approach deviates from that proposed in \cite{foi2008practical}. In \cite{foi2008practical}, the segmentation into non-overlapping level sets it performed through a wavelet decomposition of the noisy image, which we found too computationally intensive and unnecessary in our MPM context. Moreover, the protected toolbox associated with \cite{foi2008practical} is not tailored for 3D images.}
\vspace{0.5cm}

\begin{algorithm}[H]
\SetAlgoLined
\noindent\textbf{Step 1: 3D image smoothing.}
Let $\bs{y} \in \R^M$ the observed 3D image. Convolve $\bs{y}$ with a normalized uniform kernel of size $s \times s \times s$, with $s\in \N^*$, yielding a smoothed image $\bs{y}_s$.
  
\noindent\textbf{Step 2: Volume segmentation.}
\begin{enumerate}
    \item Let $J \in \N^*$. Segment $\bs{y}_s$ into $J$ distinct parts using the Lloyd-Max estimator \cite{lloyd1982least} for optimal quantization, so as to minimize the mean squared error between the segmented image and $\bs{y}_s$.
    \item Denote
    $\ell_1 \leq \ell_2 \leq \dots \leq \ell_J$ the resulting $J$ intensity levels.
    \item Denote $S_1, S_2, \dots, S_J$ the segmented parts of the image, corresponding to the intensity intervals $([\ell_{j-1}, \ell_j])_{1 \leq j \leq J}$.
\end{enumerate}  

\noindent\textbf{Step 3: Intensity and variance estimation.}
For each segmented zone $(S_j)_{1 \leq j \leq J}$:
\begin{enumerate}
    \item Estimate the intensity $\hat{I}_j$ within $S_j$ as:
    \begin{equation}
    \hat{I}_j = \frac{1}{|S_j|} \sum_{m \in S_j} y_m,
    \end{equation}
    with $|S_j|$ denotes the number of pixels (i.e., the cardinal) in $S_j$.
    
    \item Estimate the variance $\hat{\sigma}^2_j$ within $S_j$ as:
    \begin{equation}
    \hat{\sigma}_j^2 = \frac{1}{|S_j|} \sum_{m \in S_j} \left( y_m - \hat{I}_j \right)^2.
    \end{equation}
\end{enumerate}

\noindent\textbf{Step 4: Linear regression.}
Obtain the coefficients $(a, b)$ by conducting a least squares linear regression on the paired values $(\hat{I}_j, \hat{\sigma}_j^2)$ to fit Relation~\eqref{eq:linear_var}.

\caption{Estimation of parameters $a$ and $b$ \label{algo:noise_params}}
\end{algorithm}

    \subsection{A constrained parameter-free deconvolution framework}

After the noise modelling step, we move to the resolution 
of the inverse problem \eqref{eq:inverse_pb_deconvolution}. A standard approach is to minimize a cost function composed of a data fidelity term, $f\colon \R^M \longrightarrow \R $, in conjunction with a regularization function, $g\colon \R^M \longrightarrow \R$, which can be expressed as
\begin{equation}\label{eq:pb_restoration_standard}
\begin{aligned}
  &\underset{\bs{x} \in [0,+\infty)^M}{\text{minimize}} \;&&f(\bs{x}) + \chi g(\bs{x}),
\end{aligned}
\end{equation}
where $\chi>0$ denotes a regularization parameter. The data fidelity term is designed to ensure that the reconstructed data aligns with the observation model. The regularization function introduces specific \emph{a priori} characteristics to the reconstructed image, such as smoothness. 
The role of the regularization parameter is to balance the two terms. Its tuning can be based on image quality metrics such as the Mean Square Error or the SNR. However, when the ground truth or reference images are unavailable, this strategy becomes infeasible. Thus, image restoration often involves a time-consuming adjustment of the regularization parameter, with the final decision largely influenced by the user's expertise and subjective judgment. {While various statistical methods exist for estimating this parameter \cite{antoni2023bayesian,vatankhah2014regularization}, they are usually too computationally intensive for 3D data.} 

Instead, we propose to tackle the restoration challenge using a constrained approach (sometimes refereed to as the discrepancy principle). This strategy is reminiscent of those described in \cite{carrillo2012sparsity,afonso2010augmented,harizanov2013epigraphical}. The data fidelity term is constrained so that it does not exceed a known (or easily estimated) value. Our formulation reads
\begin{equation}\label{eq:pb_restoration_const}
\begin{aligned}
  &\underset{\bs{x} \in \R^M}{\text{minimize}} \;&&g(\bs{x}),\\
  & \text{subject to}
  &&f(\bs{x}) \leq B ~\text{and}~\bs{x}\in [0,+\infty)^M,
\end{aligned}
\end{equation}
where $B >0$ is a bound, either predetermined or deduced from the data. A main advantage of this formulation is that, by assuming an appropriate choice for the data fidelity function $f$, it is often feasible to derive a statistical-based upper bound for $f$, giving a good first trial value for $B$. It is worth noting that, if $f$ and $g$ are both convex, then for each $B > 0$, there usually exists a corresponding $\chi > 0$ such that \eqref{eq:pb_restoration_const} and \eqref{eq:pb_restoration_standard} are equivalent.
Let us now explicit our choices for $f$, $g$,
and $B$.

    \subsection{Choices for the data-fidelity and regularization functions}

Following our noise model \eqref{eq:noise_model}, we set
\begin{equation}
  (\forall \bs{x}\in \R^M) \quad  f(\bs{x})=\Vert \bs{W}(\bs{H} \bs{x}-\bs{y} + \alpha) \Vert^2.
\end{equation}
with $\bs{W} \in \mathcal{S}_M^+$ a weighting matrix accounting for the heterodasticity of the noise. Here, we take
\begin{equation}
\bs{W} = \mathrm{Diag}\left\{ \left( \frac{1}{\sigma([\bs{H} \bar{\bs{x}}]_m)}\right)_{1\leq m \leq M} \right\}.
\end{equation}
In practice, since {$\bs{H} \bar{\bs{x}}$ is unknown, we approximate it by $\bs{y}_s$, a denoised version of $\bs{y}$ as defined in Algorithm \ref{algo:noise_params}. }
Because of the large number of pixels,
the law of large numbers makes the following approximation valid:
$
f(\bs{x}) \approx M
$, which suggests setting $B = M$ in \eqref{eq:pb_restoration_const}.

For the regularization function $g$, we propose a re-weighted smooth total variation term, adjusted according to the voxel size along the axes $X$, $Y$, and $Z$. Specifically, given a smoothing parameter $\delta >0$, function $g$ is defined as
\begin{equation}\label{eq:def_stv}
(\forall \bs{x} \in \R^M) \quad g(\bs{x})= \sum_{m=1}^M \sqrt{ \delta + \frac{1}{r_X}(\bs{G}_X \bs{x})_m^2 + \frac{1}{r_Y}(\bs{G}_Y \bs{x})_m^2 + \frac{1}{r_Z}(\bs{G}_Z \bs{x})_m^2},
\end{equation}
with $\bs{G}= [\bs{G}_X^\top\,|\,\bs{G}_Y^\top\, |\,\bs{G}_Z^\top]^\top \in \R^{3M\times M}$,
where matrices $\bs{G}_X$, $\bs{G}_Y$, and $\bs{G}_Z$ represent discrete gradient operators along the axes $X$, $Y$, and $Z$, respectively.

    \subsection{Restoration algorithm}

The resulting optimization problem \eqref{eq:pb_restoration_const} reads as the minimization of a smoothed convex function under convex constraints. An effcient Majorization-Minimization strategy \cite{sun2017majorization},
entitled P-MMS,
was recently proposed in \cite{chouzenoux2022local} 
to address such a class of optimization problem.
Specifically tailored for large-scale constrained image processing problems such as ours, this algorithm features rapid execution times while providing theoretical convergence guarantees.

The central idea behind P-MMS
is to use the external penalty principle to cope with the constraints. Namely, for the constraints $f(\bs{x}) \leq  B$ and $\bs{x} \in [0, +\infty)^M$, we introduce corresponding penalty functions, denoted by $R_1$ and $R_2$:
\begin{equation}
    (\forall \bs{x} \in \R^M) \quad R_1(\bs{x}) = \mathrm{d}^2_{\mathcal{B}(\bs{0}_M, B)} \left(\bs{W}(\bs{H} \bs{x}-\bs{y}+\alpha)\right),
\end{equation}
and 
\begin{equation}
     (\forall \bs{x}\in \R^M) \quad R_2(\bs{x}) = \mathrm{d}^2_{[0, +\infty)^M} \left(\bs{x}\right),
\end{equation}
where $\mathrm{d}_C$ is the distance to the set $C \subseteq \R^M$. Note that the gradient of $R_1$ and $R_2$ are easily computable since the projections onto the ball $\mathcal{B}(\bs{0}_M, B)$ and the nonnegative orthant $[0, +\infty)^M$ are closed form.
A key assumption for P-MMS
to be applicable is that every function in the problem admits a tangent quadratic upper bound at every point. 
In mathematical terms, for a differentiable function $\psi\colon \R^M \longrightarrow \R$, this requirement corresponds to the existence, for every $\bs{x}' \in \R^M$, of a curvature matrix $\bs{A}_\psi(\bs{x}') \in \mathcal{S}_M^+$, such that
\begin{equation}
(\forall \bs{x}\in \R^M)\quad \psi(\bs{x}) \leq \psi(\bs{x}') + \nabla \psi(\bs{x}')^\top (\bs{x}-\bs{x}') + (\bs{x}-\bs{x}')^\top \bs{A}_\psi(\bs{x}') (\bs{x}-\bs{x}').
\end{equation}
The existence of such quadratic upper bounds can be readily derived for each function $g$, $R_1$, and $R_2$ in our problem, using the descent Lemma \cite[Lemma 2.64]{bauschke2019convex} or the half-quadratic majorization formula \cite[Lemma 1]{chouzenoux2013majorize}. This leads the following valid curvature matrices:
\begin{equation}\label{eq:def_A_Psi}
\bs{A}_g(\bs{x}) = \bs{G}^\top \mathrm{BDiag}\left\{ \left( \left(\delta + \frac{(\bs{G}_X \bs{x})_m^2}{r_X} + \frac{(\bs{G}_Y \bs{x})_m^2}{r_Y} + \frac{(\bs{G}_Z \bs{x})_m^2}{r_Z} \right)^{-1/2} \bs{I}_3 \right)_{1\leq m \leq M} \right\} \bs{G},
\end{equation}
\begin{equation}
    \bs{A}_{R_1}(\bs{x}) = \bs{H}^\top   \bs{W}^\top \bs{W}\bs{H}, \quad \text{and} \quad  \bs{A}_{R_2}(\bs{x}) = 2 \bs{I}_M,
\end{equation}
where $\mathrm{BDiag}$ stand for block diagonal matrix. Then, for any penalty parameter $\gamma>0$, the penalized function
\begin{equation}
    F_{\gamma}(\bs{x}) = g(\bs{x}) + \gamma (R_1(\bs{x}) + R_2(\bs{x}))
\end{equation}
admits a quadratic tangent majorant at $\bs{x}$ uniquely defined by its curvature matrix:
\begin{equation}\label{eq:def_A_global}
    \bs{A}_{F_\gamma}(\bs{x}) = \bs{A}_{g}(\bs{x}) + \gamma (\bs{A}_{R_1}(\bs{x}) + \bs{A}_{R_2}(\bs{x})).
\end{equation}
We present in Algorithm~\ref{algo:P-MMS} the P-MMS approach, to solve Problem~\eqref{eq:pb_restoration_const}. At each iteration of the inner algorithm, $\mathbf{x}_{k+1}$ is updated within the affine space defined by the directions $\{-\nabla F_{\gamma_j}(\mathbf{x}_k), \mathbf{x}_k - \mathbf{x}_{k-1}\}$, employing a Majorization-Minimization approach to determine the multidimensional step-size.
A local version P-MMS$^{\text{loc}}$ of this algorithm was developed which implements a trust-region strategy to accelerate further convergence. We redirect readers to \cite{chouzenoux2022local} or to its publicly available code, for more details. The P-MMS algorithm benefits from the convergence guarantees given in Theorem \ref{thm:convergence_PMMS}.

\begin{algorithm}
\caption{P-MMS \label{algo:P-MMS}}
Inputs: $(\gamma_j)_{j\in \N} \in (\R^+)^{\N}$, $(\varepsilon_j)_{j\in \N} \in  (\R^+)^{\N}$, $\bs{x}_0 \in \R^M$.\\
\For{$j=0,1, \dots$}{
\tcp{find an (approximated) minimizer of $F_{\gamma_j}$}
Set initial point $\mathbf{x}_0$,\\
\For{$k=1, \dots$}{
Construct subspace directions $\bs{D}_k = [-\nabla F_{\gamma_j}(\mathbf{x}_k), \mathbf{x}_k - \mathbf{x}_{k-1}]$,\\
Compute $\bs{A}_{F_{\gamma_j}}(\mathbf{x}_k)$ according to \eqref{eq:def_A_global},\\
$\bs{u}_{k}= -\left[\bs{A}_{F_{\gamma_j}}(\mathbf{x}_k)\right]^{\dagger} \bs{D}_k^\top \nabla F_{\gamma_j}(\mathbf{x}_k)$,\\
$\mathbf{x}_{k+1}= \mathbf{x}_k + \bs{D}_k \bs{u}_{k}$,\\
\If{$\Vert \nabla F_{\gamma_j}(\mathbf{x}_{k+1}) \Vert < \varepsilon_j$}{exit loop \tcp{stop inner algorithm if given precision $\varepsilon_j$ on the norm of the gradient is reached}
return $\mathbf{x}_{k+1}$.}}
$\bs{x}_j = \mathbf{x}_{k+1}.$}
\textbf{return } $\bs{x}_{j}$.
\end{algorithm}

\begin{theorem}[Convergence of the P-MMS algorithm]\label{thm:convergence_PMMS}
Assume the sequence of parameters $(\varepsilon_j)_{j\in \N}$ satisfies, for every $j \in \N$, $\varepsilon_j >0$ and $\lim_{j \rightarrow +\infty} \varepsilon_j = 0$. Also assume that $(\gamma_j)_{j\in \N}$ is a nondecreasing sequence of positive reals and $\lim_{j \rightarrow +\infty} \gamma_j = +\infty$.
Then, the sequence $(\bs{x}_j)_{j\in \N}$ generated by Algorithm \ref{algo:P-MMS} is bounded and any of its cluster point is a solution to Problem \eqref{eq:pb_restoration_const}.
\end{theorem}

\begin{proof}
    It suffices to check that the assumptions for \cite[Theorem 2]{chouzenoux2022local} are satisfied. 
    \begin{itemize}
        \item The functions $g$, $R_1$ and $R_2$ are differentiable,
        \item For every $\gamma >0$, the function $F_\gamma$ is coercive, convex, and satisfies the K\L{} property.
        \item For every $\gamma >0$, the curvature function defined in \eqref{eq:def_A_global} is lower bounded independently of $\bs{x}$ and is continuous.
    \end{itemize}
This concludes the proof.
\end{proof}

    \subsection{Validation of the restoration method on simulated data}\label{sec:val_restoration}

To illustrate the performance of our heteroscedastic constrained formulation for the MPM image restoration tast, we first conducted an experiment using simulated data. 

For this experiment, we chose an image of a fly brain from \cite{chalvidal2023block} as the object of interest $\bar{\bs{x}}$. This image, with dimensions $M=128\times 128 \times 40$, was artificially degraded through a convolution operator $\bs{H}$ mimicking the effect of a normalized 3D Gaussian kernel with inverse covariance matrix parameterized by the angles $\theta=5\pi/6$ rad, $\varphi=0$ rad, and eigenvalues $\bs{s} = (50, 50, 20)$. Moreover, heteroscedastic noise was introduced based on the noise model presented in \eqref{eq:noise_model}, with $a=0.01$ and $b=10^{-5}$.
The values for the noise and the blur parameters were chosen according to the typical magnitudes one may encounter in MPM. In MPM, noise levels are inherently high, and the PSF width is significant, especially in the depth direction. The voxels dimensions were configured to $0.05 \times 0.05 \times 0.05\mu$m$^3$.

To provide a comprehensive evaluation, we compare our approach against the more traditional regularized approach in \eqref{eq:pb_restoration_standard}, where we choose the data-fidelity function $f$ to fit homoscedastic Gaussian noise, setting $f(\bs{x})= \Vert \bs{H} \bs{x}-\bs{y} + \alpha\Vert^2$, and the smoothed TV regularization function $g$ defined in \eqref{eq:def_stv}.

For both models, the restoration was performed running the P-MMS algorithm. The parameter $\delta$ was set to $0.1$ and, for the constrained problem, the bound $B$ to the number of pixels, i.e., $B=M$. 
The penalty parameters $(\gamma_j)_{j\in \N}$ in Algorithm \ref{algo:P-MMS} were set to $\gamma_j = (2j)^{2}$ and the precision parameters $(\varepsilon_j)_{j\in \N}$ were chosen as $\varepsilon_j = 10^5\big/(\gamma_j)^{0.75}$. The algorithm was initialized with $\bs{x}_0 = \bs{y}$. Using the same configuration as previously (a desktop machine having an Intel i7-4790, 4 CPU cores, and 16 GB of RAM), the deconvolution on the synthetic fly brain image takes $2$ min and $10$ seconds to reach the final SNR (up to 2 digits).

Figure \ref{fig:convergence_check} displays the SNR as a function of the iteration number of the P-MMS algorithm for the constrained formulation, confirming the theoretically convergence established in Proposition \ref{thm:convergence_PMMS}, and illustrating the fast stability after few dozens of iterations only. Figure \ref{fig:comparison_reg_const} captures the performance of the two methodologies. It showcases the SNR of the restored images using the standard regularized approach as a function of the regularization parameter $\chi$, as well as the SNR achieved using our parameter-free method. A notable observation is the important fluctuation in SNR values for the regularized method as $\chi$ gets closer to its optimal value. It is crucial to highlight that this optimal value of $\chi$ remains unknown in the real-world context of Multiphoton Microscopy (MPM), as the SNR metric is not available. In contrast, our constrained methodology consistently delivers superior SNR, underscoring the interest of a reliable noise model. The absence of any tuning parameter is furthermore a crucial advantage in real-life applications.
Finally, Figure \ref{fig:simulation_restorations} complements our quantitative analysis with visual comparisons. It displays the 10th slice in the XY plane of (a) the original image, (b) the degraded image with blur only, (c) the degraded image with blur and noise, and the reconstructions obtained using (d) the regularized least-squares approach for the value of $\chi$ maximizing the output SNR, (e) Richardson-Lucy's deconvolution method \cite{richardson1972bayesian} and (f) our proposed constrained approach assuming Gaussian heteroscedastic noise. Our approach yields superior visual restoration compared to the commonly used Richardson-Lucy deconvolution method in microscopy \cite{ingaramo2014richardson,holmes1989richardson} and achieves similar visual results to the regularized least-squares approach, with the added advantage of being parameter-free.

\begin{figure}
    \centering
\includegraphics[scale=0.5]{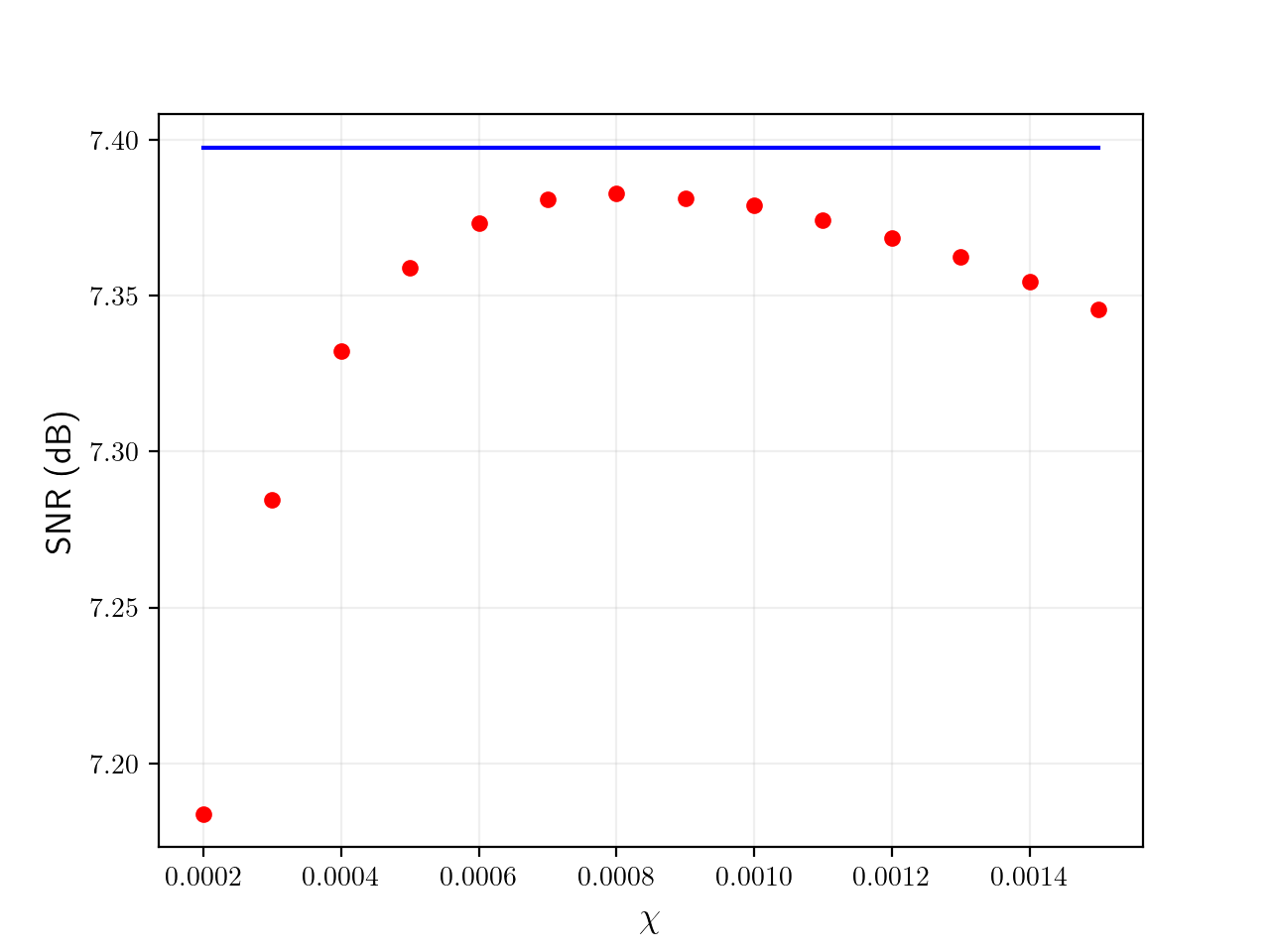}
    \caption{Comparison of the penalized least-squares approach (red dots) with the proposed constrained approach assuming a Gaussian heteroscedastic noise model (blue line), in terms of restored image SNR, in dB.}
    \label{fig:comparison_reg_const}
\end{figure}

\begin{figure}
    \centering
    \includegraphics[scale=0.5]{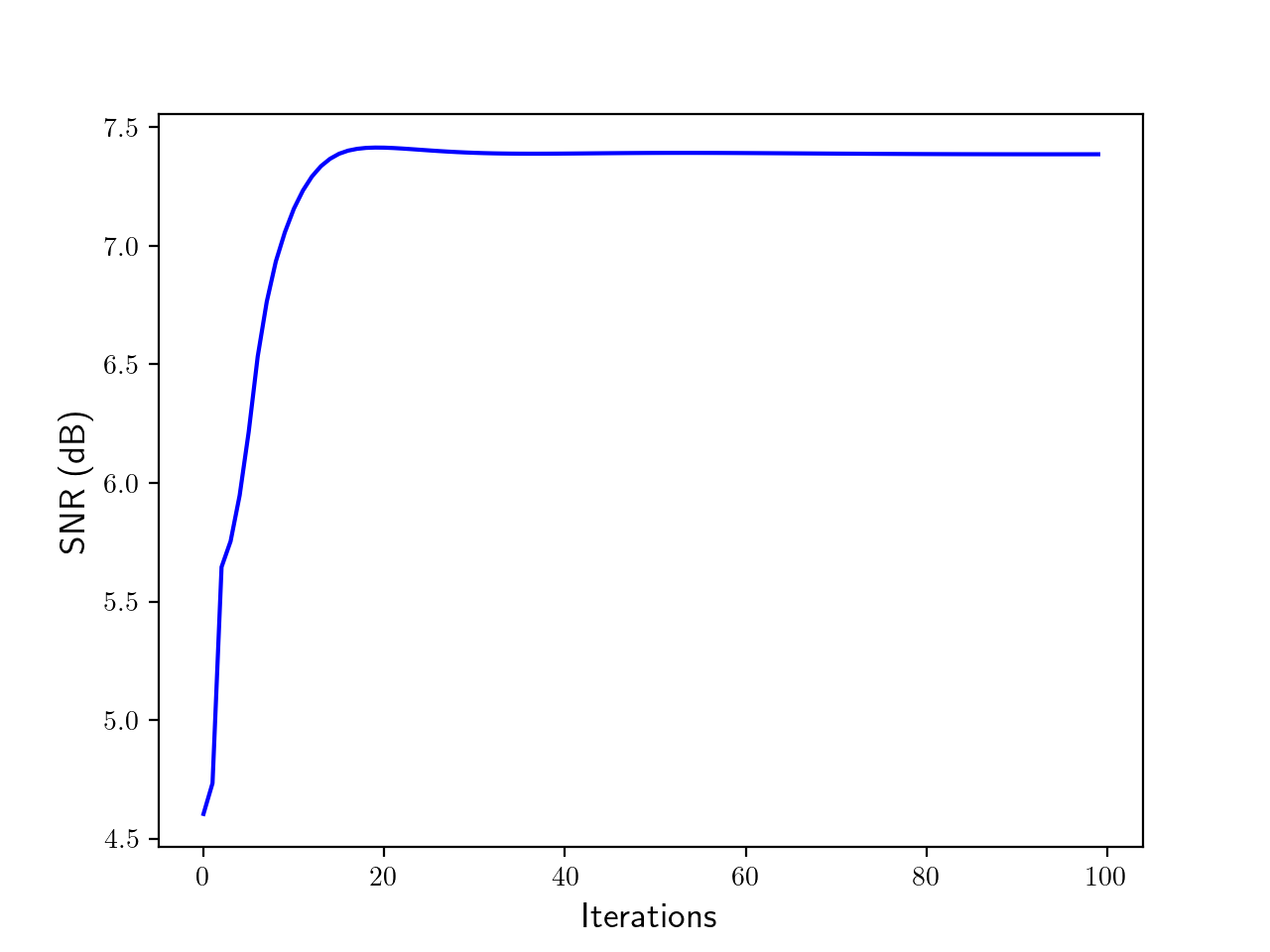}
    \caption{SNR evolution along the iterations $j$ in Algorithm \ref{algo:P-MMS} for the resolution of the proposed constrained formulation.}
    \label{fig:convergence_check}
\end{figure}

\begin{figure}
\centering
\begin{minipage}[c]{0.3\textwidth}
  \centering
 \includegraphics[width=1.0\textwidth]{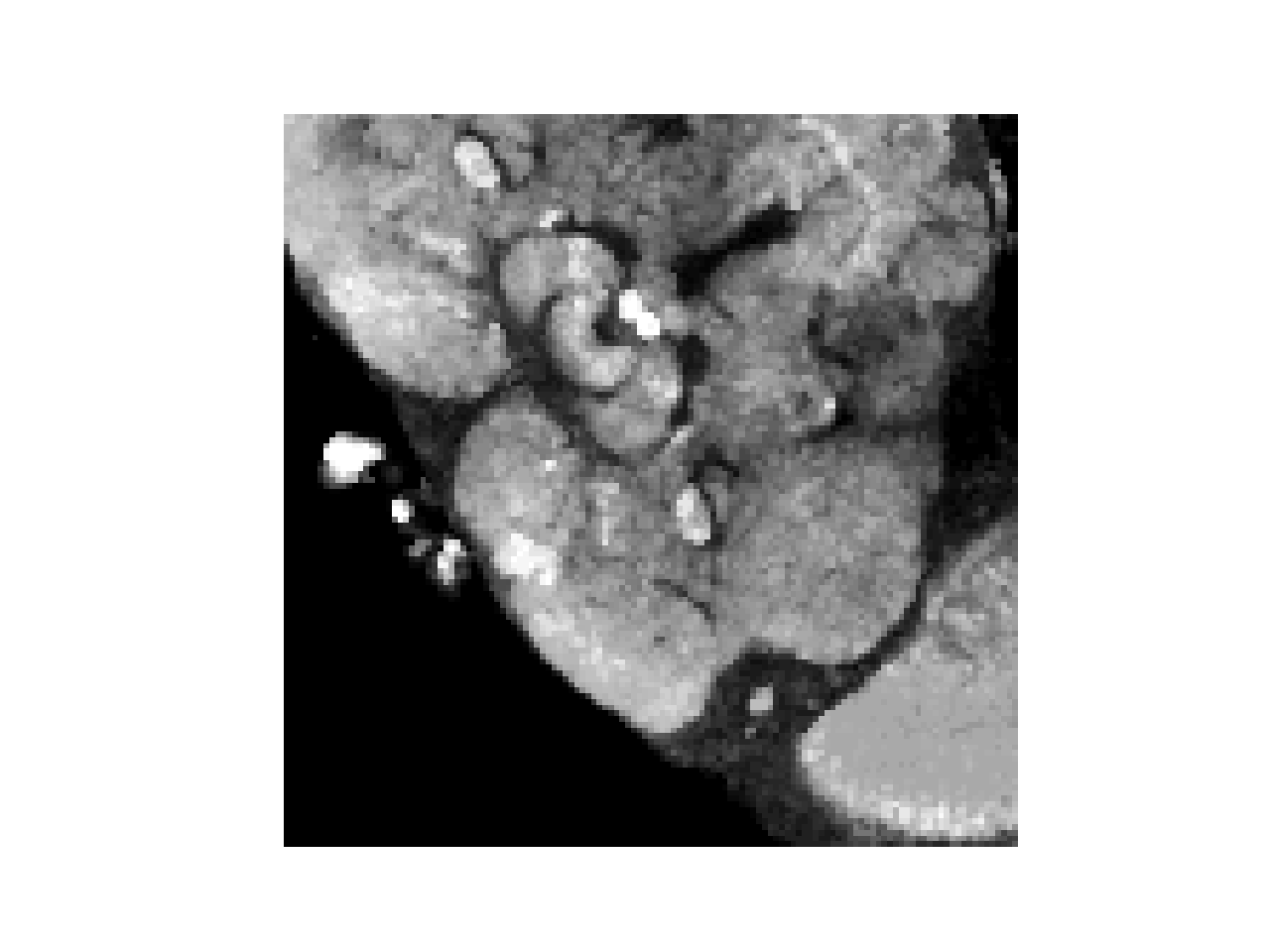} \subcaption[]{Original  \\  ~ \label{fig:1a}}
\end{minipage}%
\begin{minipage}[c]{0.3\textwidth}
  \centering
 \includegraphics[width=1.0\textwidth]{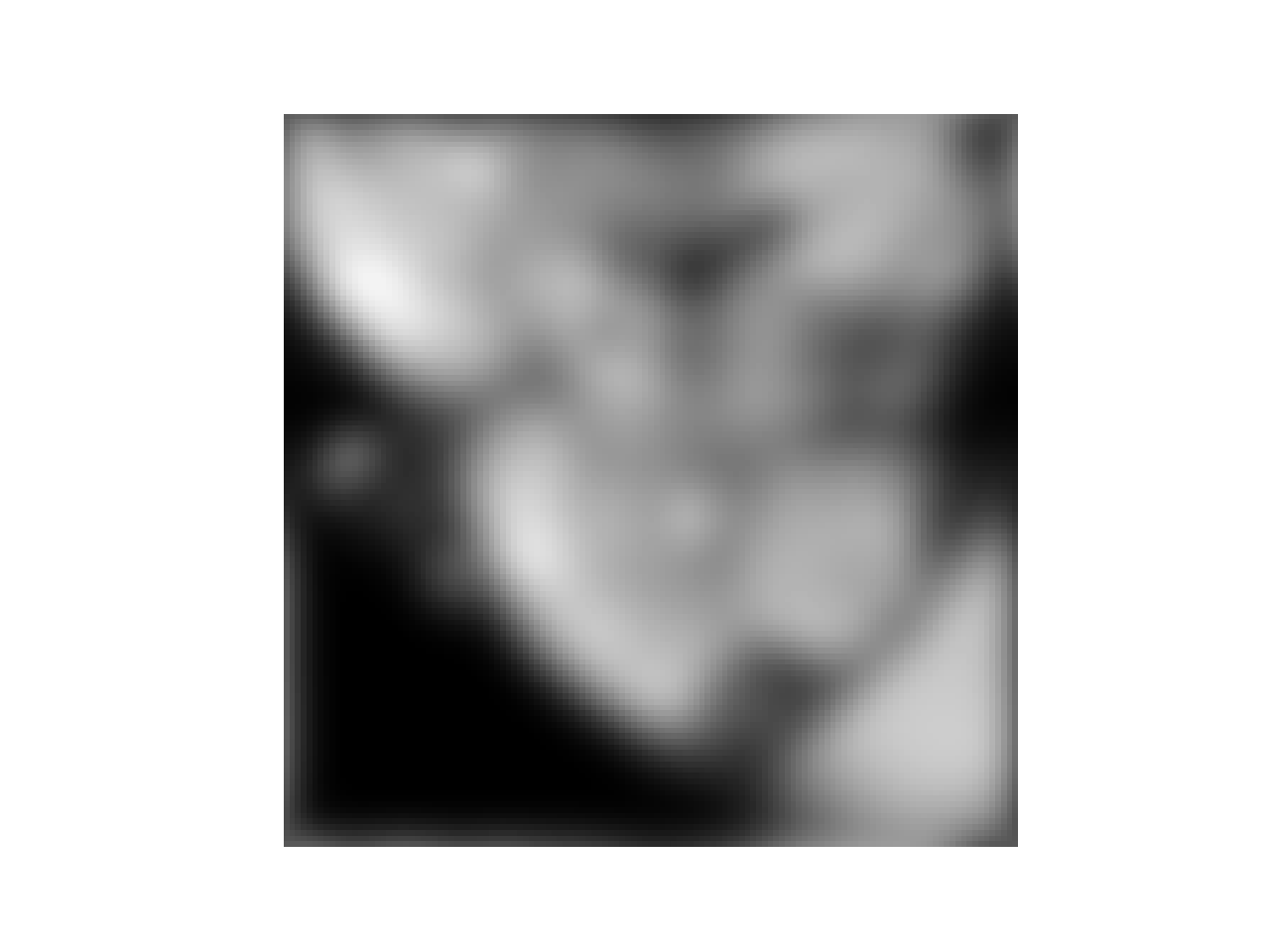} \subcaption[]{Degraded with blur \\ (SNR=$5.79$dB) \label{fig:1b}}
\end{minipage}
\begin{minipage}[c]{0.3\textwidth}
  \centering
 \includegraphics[width=1.0\textwidth]{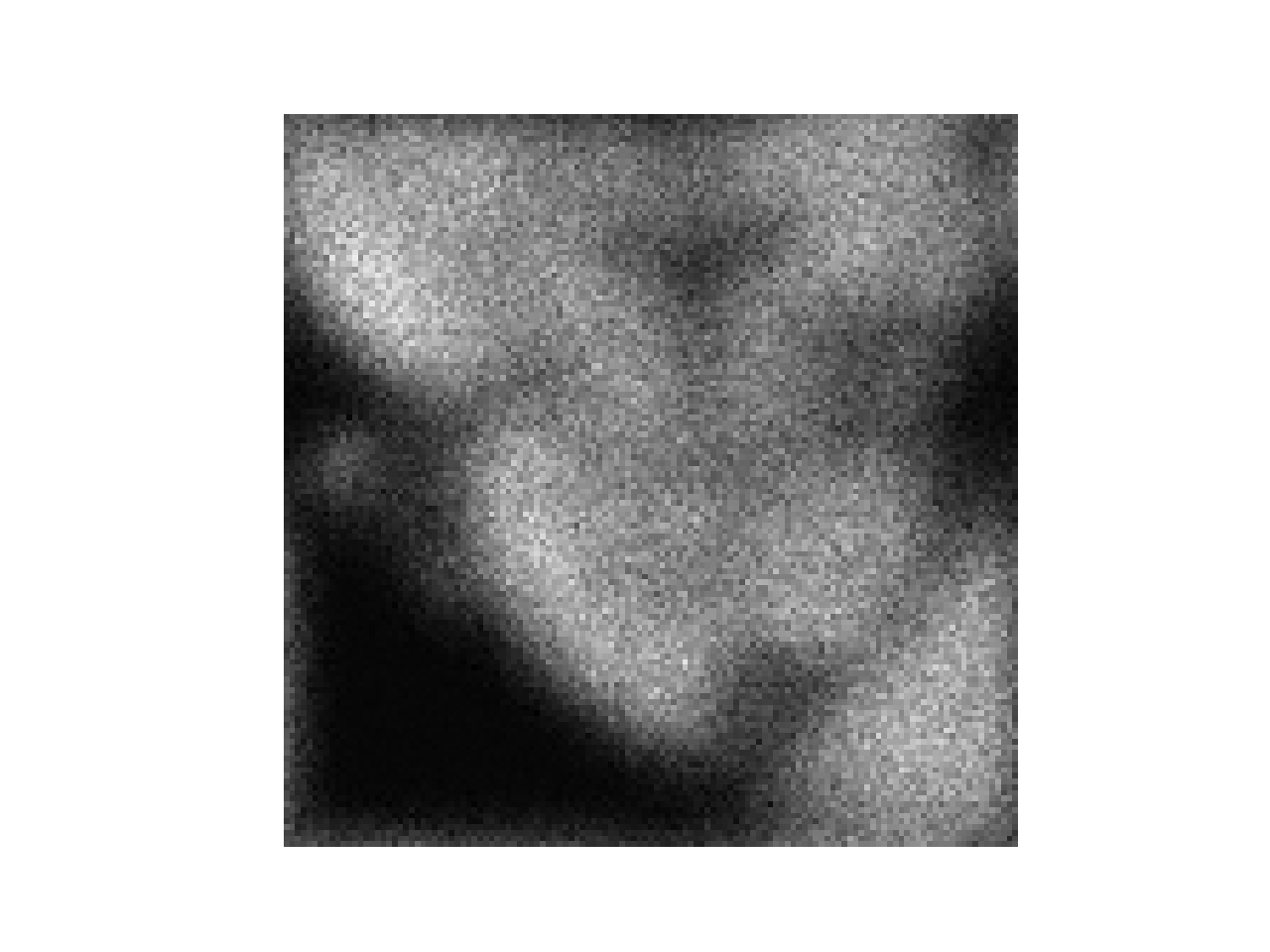} \subcaption[]{Degraded with blur \\ and noise (SNR=$5.26$dB) \label{fig:1c}}
\end{minipage}%

\begin{minipage}[c]{0.3\textwidth}
  \centering
 \includegraphics[width=1.0\textwidth]{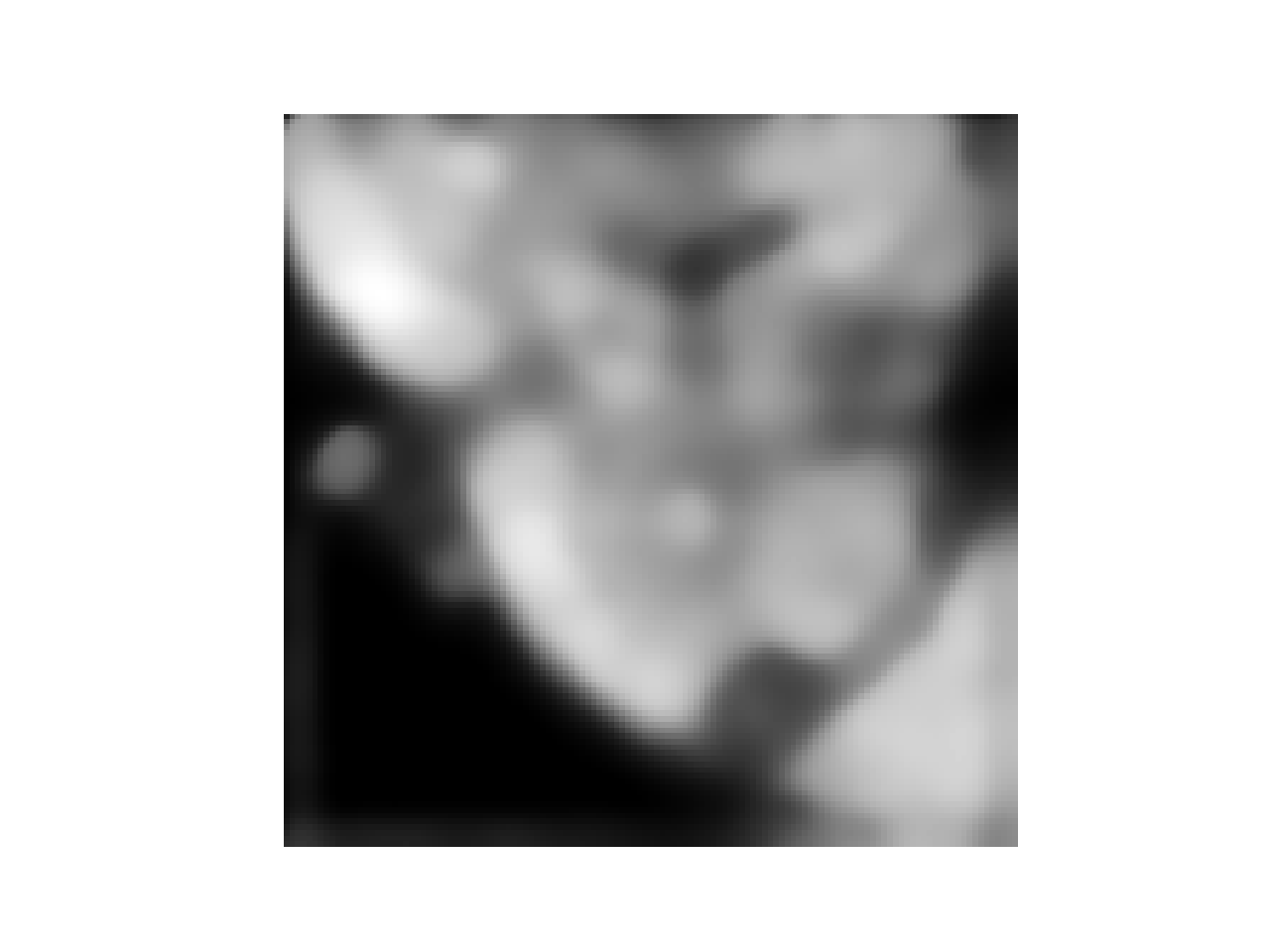} \subcaption[]{Regularized least\\-squares (SNR=$7.38$dB) \label{fig:1d}}
\end{minipage}%
\begin{minipage}[c]{0.3\textwidth}
  \centering
 \includegraphics[width=1.0\textwidth]{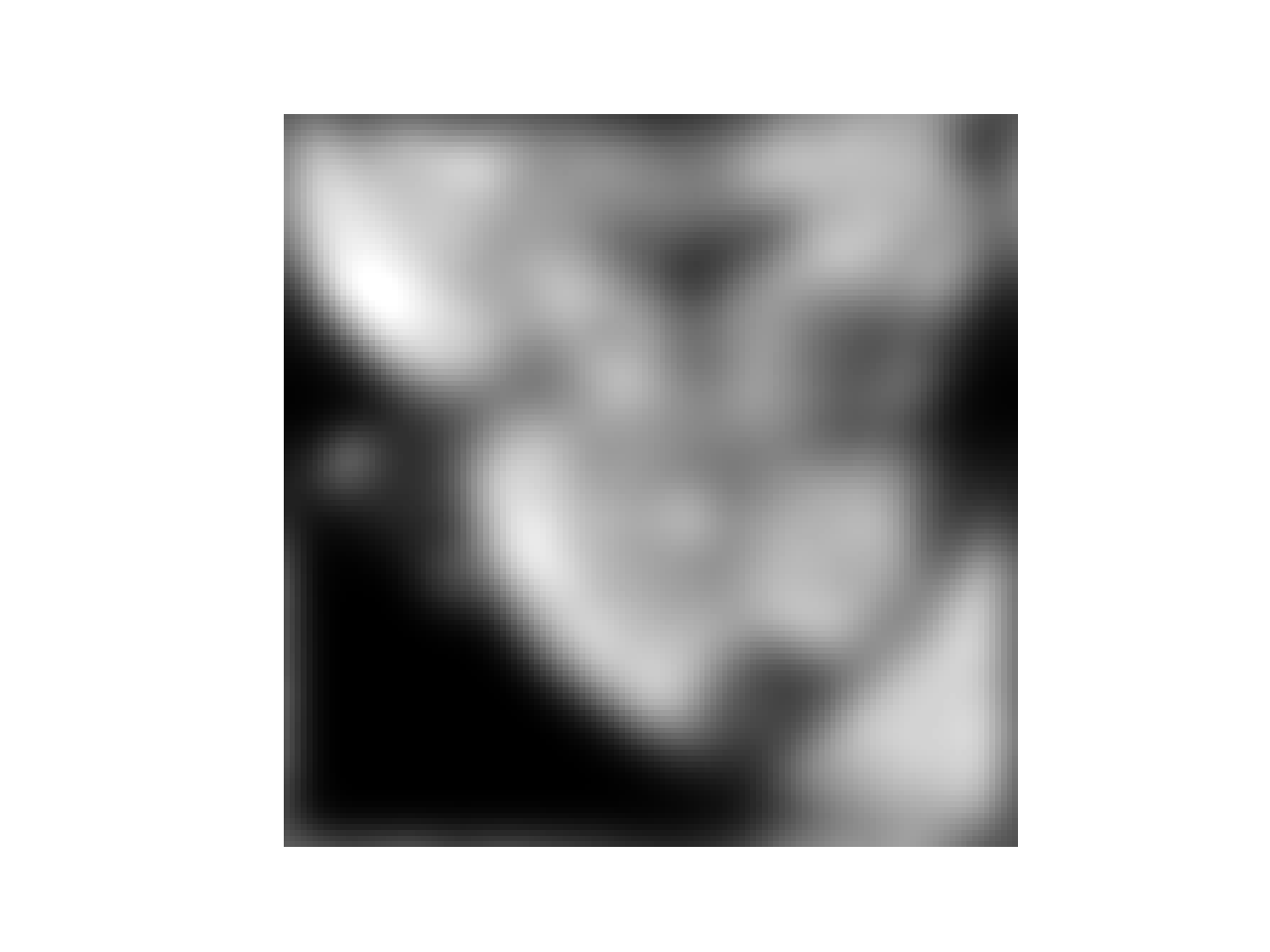} \subcaption[]{Richardson-Lucy \\(SNR=$5.92$dB) \label{fig:1e}}
\end{minipage}%
\begin{minipage}[c]{0.3\textwidth}
  \centering
 \includegraphics[width=1.0\textwidth]{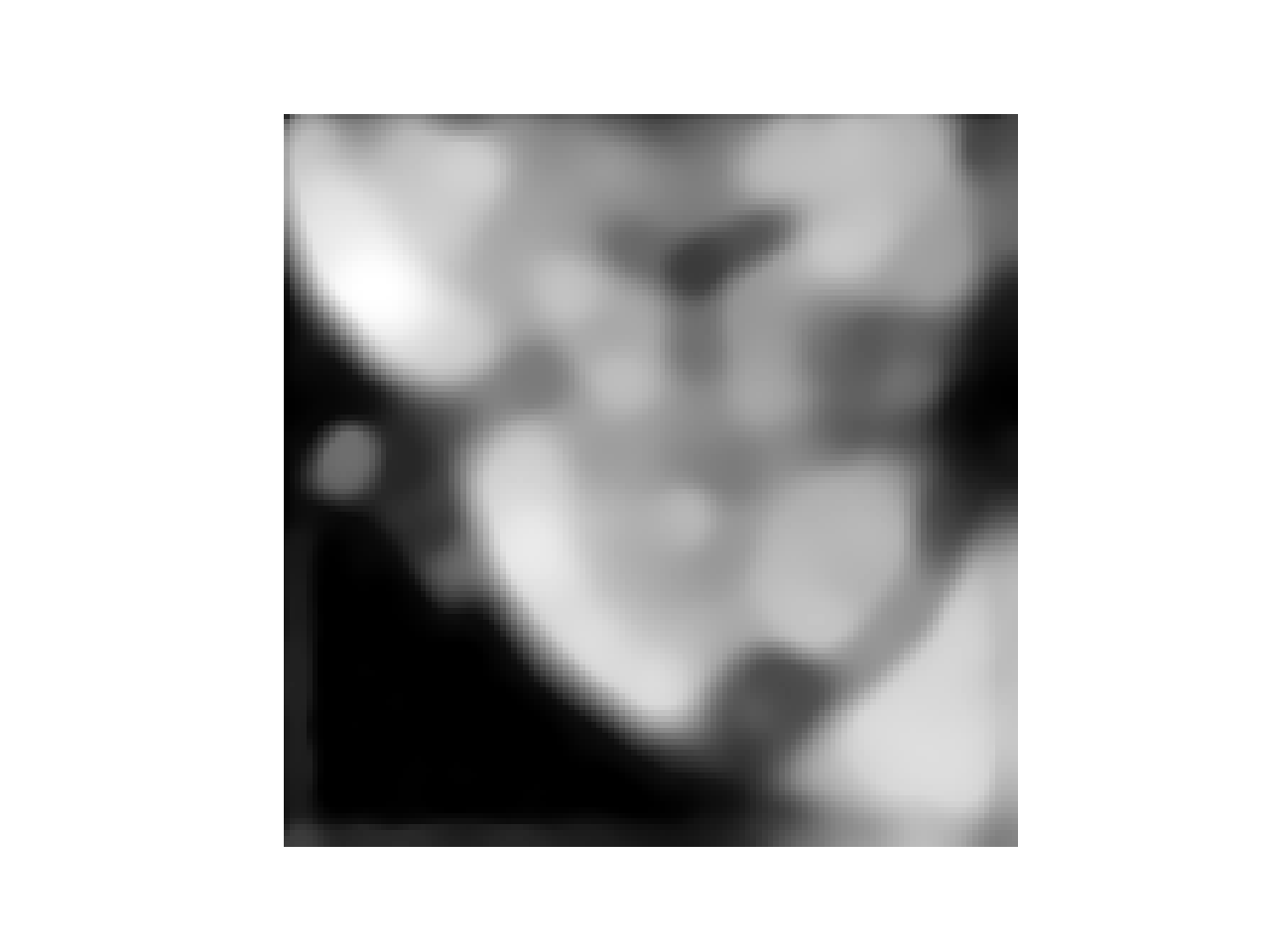} \subcaption[]{Ours \\(SNR=$7.40$dB) \label{fig:1f}}
\end{minipage}%
    \caption{Deconvolution results for the synthetic fly brain image. Slice 10-th is displayed. SNR values are computed on the whole volume.}
    \label{fig:simulation_restorations}
\end{figure}

\section{Application of the proposed pipeline to muscle tissue imaging}
\label{sec:mouse}
In this section, we apply the entire pipeline, encompassing both PSF estimation and deconvolution, to the restoration of a mouse muscle volume, characterized by its prominent myosin filaments, from MPM acquisitions. {All mice were bred and housed in Limoges University’s animal facility under controlled conditions (20$^\circ$C, 12 hours light/12 hours dark cycle) with free access to standard mouse chow and tap water. Experimental procedures were carried out in accordance with the European legislation on animal experimentation (Directive 2010/63/UE), and approved by Ethical Committee no.033 (APAFIS \#1903-2015091612088147 v2). Phenotypic and molecular analyses were performed on 12-week-old male animals.}

    \subsection{Presentation of the experimental setup}

A mouse muscle sample was immersed in a water solution. As a reference for the PSF estimation phase, polystyrene microspheres of $1 \mu \text{m}$ diameter were integrated along the muscle's perimeter. We employed a bi-channel raw acquisition technique, combining both the Second Harmonic Generation signal from the myosin and the two-photon fluorescence from the microspheres. Thus, the PSF is calibrated in the same field of view, and under the same resolution, than the one where we observe the myosin organization. Such experimental setting is somewhat ideal, as the conditions for recording the PSF match precisely those for capturing the image of the muscle sample. Given that the muscle assembly acts as an absorbing and scattering medium that distorts the optical wavefront, recording the PSF in the location where the image is restored enhances restoration accuracy. The microscopy acquisition was carried out with an excitation wavelength of $810 \text{nm}$. We used the same device than in our previous experiments, that is a water immersion objective tailored for multiphoton acquisitions, specifically the Olympus XLPLN25XWMP with a numerical aperture $ \text{NA} = 1.05$. The selected voxel dimensions for the imaging process were $0.049 \times 0.049 \times 0.1 \mu \text{m}^3$. The resulting volume was of dimension $M = 840 \times 840 \times 180$.

    \subsection{Results for PSF calibration step}

We first conducted the PSF estimation on the channel of the acquired image that contains the beads, that is we perform an \emph{in situ} PSF calibration within an heterogeneous medium. Given these conditions, we anticipate a notably higher dispersion compared to the optimal setting discussed in Section \ref{subsubsec:exp_beads_opt}.
We fitted a convolution kernel on 8 cropped volumes, containing isolated bead images, using GENTLE (Algorithm \ref{algo:MPM_PSF}). The hyper-parameters of the PSF estimation model \eqref{eq:minimization_pb_PSF} were set identically to Section \ref{sec:validation_PSF}.
Figure \ref{fig:FWHM_muscle} presents a comparison of the estimated FWHMs along the principal axis and Euler angles for the 8 beads, using the GENTLE method.
Estimations across different beads again exhibit consistency, confirming the PSF stationarity assumption in this experiment.
 In addition, we display in Figure \ref{fig:intensity_profiles} the intensity profiles of the observed beads and of their reconstitution $\hat{\alpha} + \hat{\beta}\hat{\boldsymbol{h}} \ast \boldsymbol{x}$, where $\boldsymbol{x}$ is the theoretical bead. The GENTLE method accurately fits the observation across all three axes and, in particular, fits asymmetrical shapes more effectively than NLS, thanks to its robust formulation. 
\begin{figure}[H]
\centering
\includegraphics[width=0.4\linewidth]{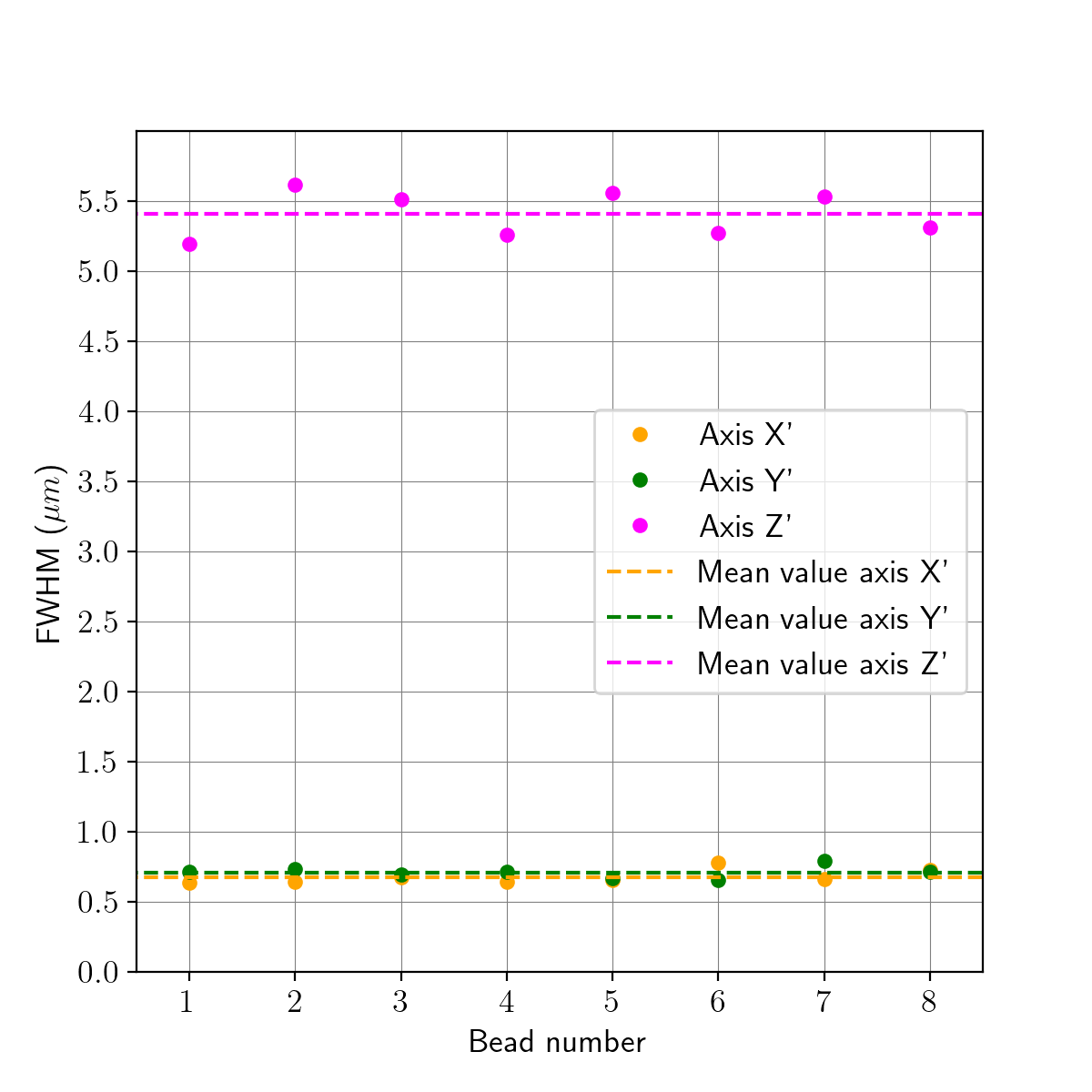} 
\includegraphics[width=0.4\linewidth]{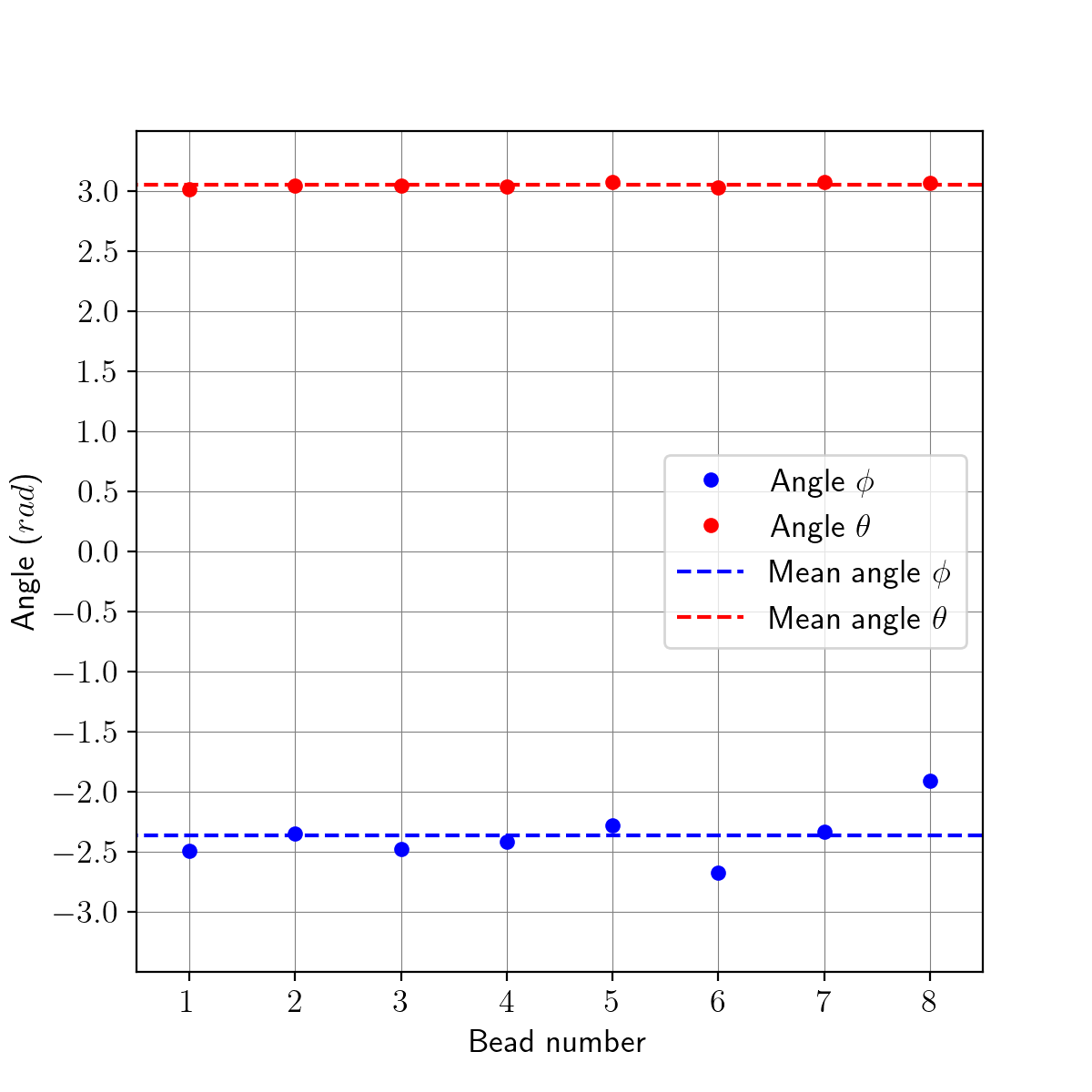}
\caption{Estimated FWHM (left) in the 3 principal component directions denoted by $X'$, $Y'$ and $Z'$, and angles (right) for the eight isolated beads.}
\label{fig:FWHM_muscle}
\end{figure}

\begin{figure}[H]
\centering
\includegraphics[scale=0.3]{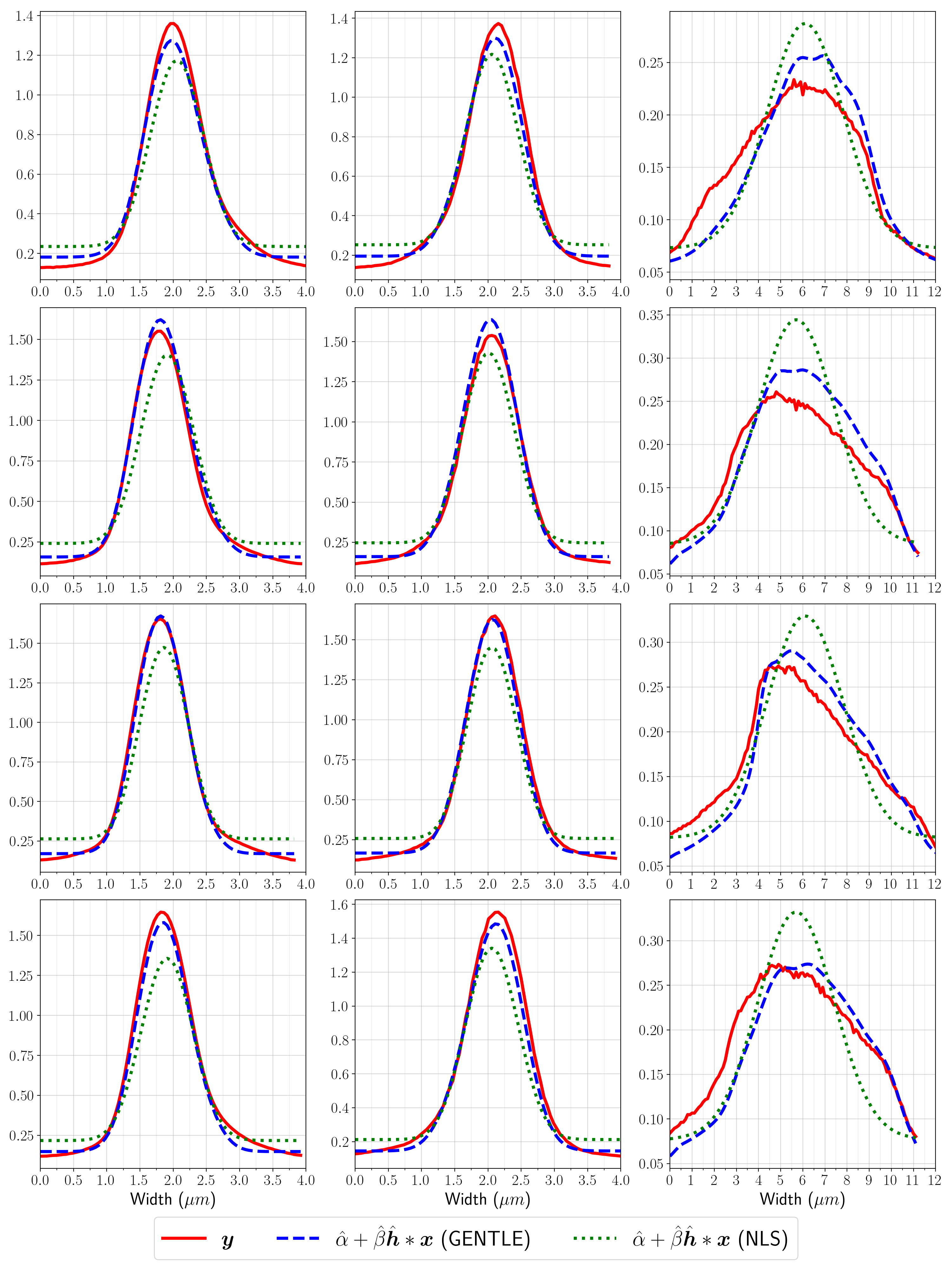}
\caption{Profiles along axis X (left), Y (center) and Z (right) for 4 of the 8 isolated beads (rows). We compare the intensity profiles of the observation $\boldsymbol{y}$ (red), the reconstruction $\hat{\alpha} + \hat{\beta}\hat{\boldsymbol{h}} \ast \boldsymbol{x}$ obtained with GENTLE (blue) and the reconstitution obtained with NLS (green). The profiles for the remaining 4 beads, omitted for space reason, are consistent with the displayed ones.
}\label{fig:intensity_profiles} 
\end{figure}

    \subsection{Results for the restoration step}
In the restoration phase, we first determined the noise parameters $a$ and $b$ across the entire volume using the methodology described in Algorithm \ref{algo:noise_params} with $s=5$ and $J=25$. Figure \ref{fig:noise_params_muscle} displays the linear regression corresponding to the fourth step of the algorithm. The close alignment of the points confirms the appropriateness of the heteroscedastic noise model proposed.

Subsequently, the standard deviations $(\sigma([\bs{H} \bar{\bs{x}}]_m))_{1 \leq m \leq M}$ were approximated using \eqref{eq:linear_var} with the estimated $a$, $b$, and a smoothed version of the image $\bs{H} \bar{\bs{x}}$ corresponding to the convolution of $\bs{y}$ with a $3\times 3 \times 3$ uniform kernel. 
We solved Problem \eqref{eq:pb_restoration_const} by setting $B = M$ and $\delta=0.1$, using Algorithm \ref{algo:P-MMS}. The penalty parameters $(\gamma_j)_{j\in \N}$ were set to $\gamma_j = (1.5j)^{1.2}$ and the precision parameters $(\varepsilon_j)_{j\in \N}$ were chosen as $
\varepsilon_j = 10^5\big/(\gamma_j)^{0.5}$.

In Figure \ref{fig:restoration_muscle}, we display slices from a section of the raw image alongside corresponding slices from the restored images obtained using GENTLE, as well as with various state-of-the-art methods described hereafter. Note that comparisons include only 3D non-blind, or 2D blind deconvolution approaches, due to the lack of publicly available 3D blind methods.
\begin{itemize}
\item The Huygens software, which is commercially licensed and popular in the MPM community. This software performs deconvolution of 3D volumes based on the theoretical PSF, as outlined in Equations \eqref{eq:FWHM_theo_XY} and \eqref{eq:FWHM_theo_Z} in our paper.
\item BlindDeconv, the blind deconvolution method from \cite{krishnan2011blind}, designed primarily for 2D images and applied here to each 2D slice in the XY plane.
\item The Richardson-Lucy deconvolution algorithm from \cite{richardson1972bayesian}, applied using the PSF estimated by our method.
\item SplitGrad \cite{chouzenoux2012majorize}, a regularized adaptation of the Richardson-Lucy algorithm, run with our estimated PSF.
\end{itemize}
For methods requiring a regularization parameter (namely, BlindDeconv and SplitGrad), we selected the best output based on subjective visual improvement, as there is no ground truth for our dataset. These results are shown in Figure \ref{fig:restoration_muscle}. 

The all-in-one restoration with the commercial Huygens software appears unconvincing because neither the noise nor the blur has been effectively reduced. As expected, the 2D blind deconvolution method struggles with deconvolution in the Z-direction. While Richardson-Lucy and SplitGrad yield visually reasonable outcomes. In contrast with the image restored with our GENTLE method, their ability to clearly reveal myofibrils (round structures at the fiber level) is limited. Additionally, Figure \ref{fig:3D_visualization} provides a 3D visualization of both the raw and restored images, and Figure \ref{fig:profiles_z_restoration} shows their profiles along the Z-axis, demonstrating effective denoising and deconvolution by our GENTLE pipeline.

\begin{figure}
    \centering
    \includegraphics[scale=0.5]{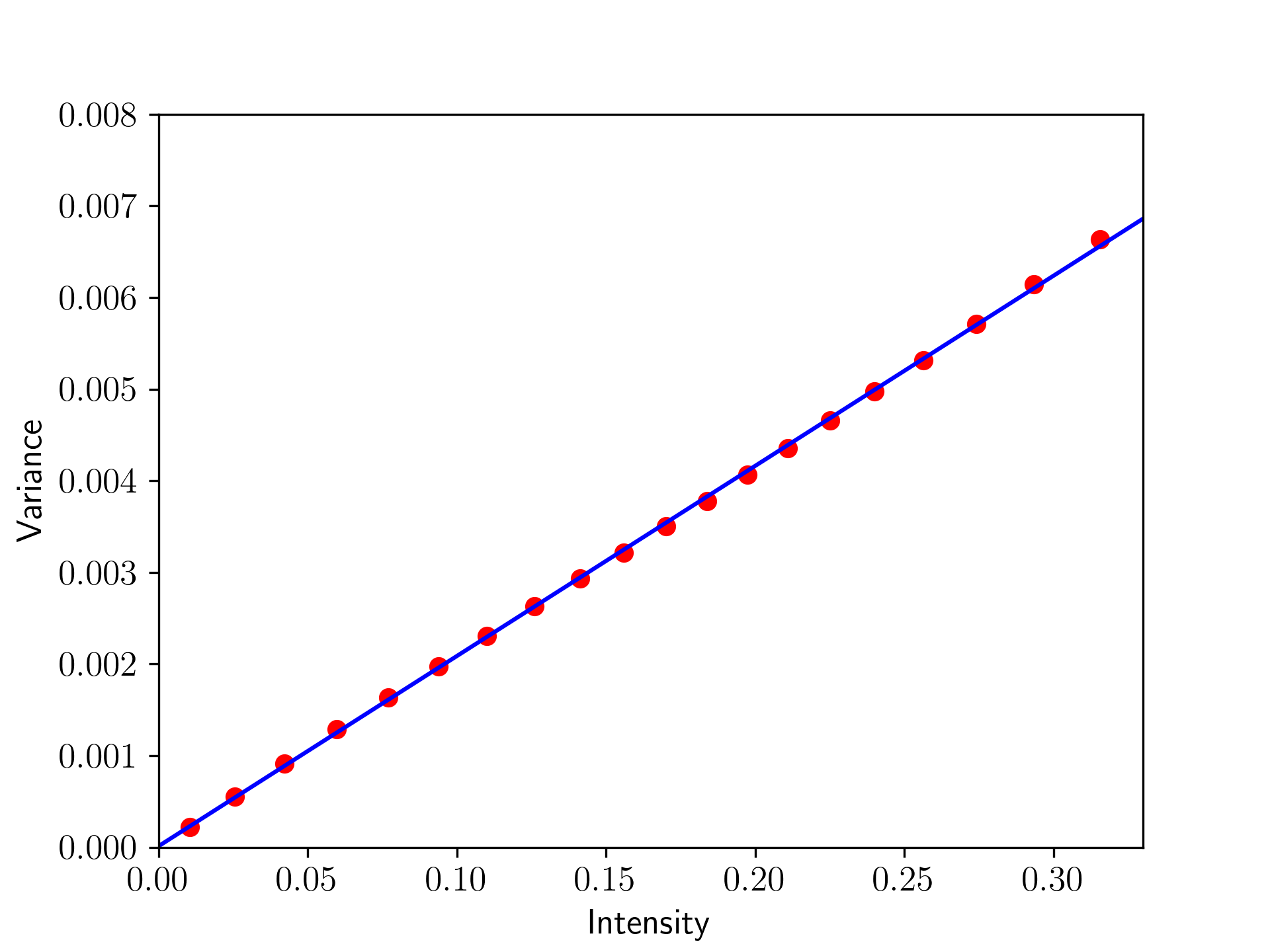}
    \caption{Linear regression (blue line) performed on the couples of values $(\hat{I}_j, \hat{\sigma}_j^2)$ (red dots) for the estimation of noise parameters $(a, b) \in \R^2$ on the mouse muscle 3D image. }
    \label{fig:noise_params_muscle}
\end{figure}


\begin{figure}
    \centering
        \begin{tabular}{@{}l@{}l@{}}
           \includegraphics[scale=0.3]{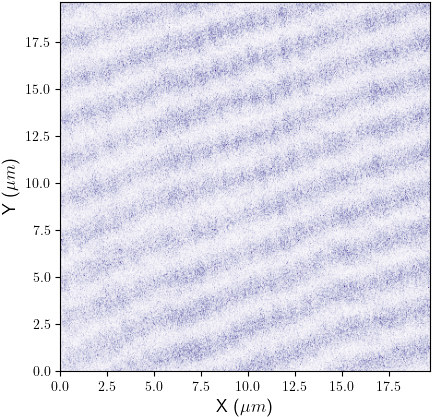}  & \includegraphics[scale=0.3]{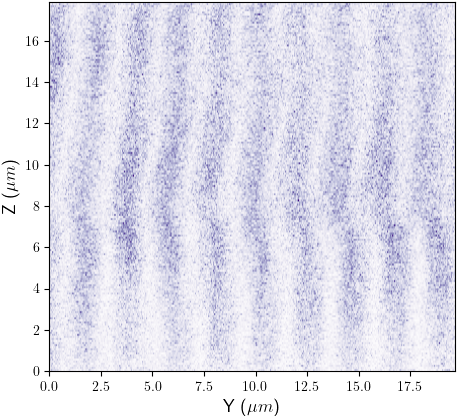}   \\
           \includegraphics[scale=0.3]{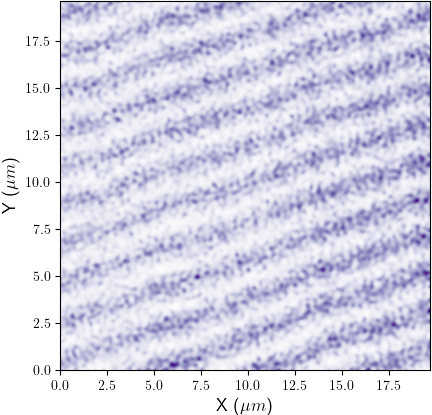} &  \includegraphics[scale=0.3]{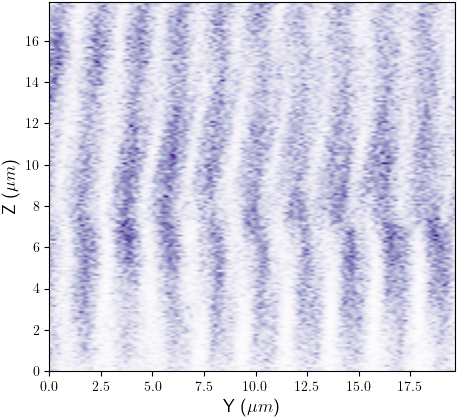} \\
           \includegraphics[scale=0.3]{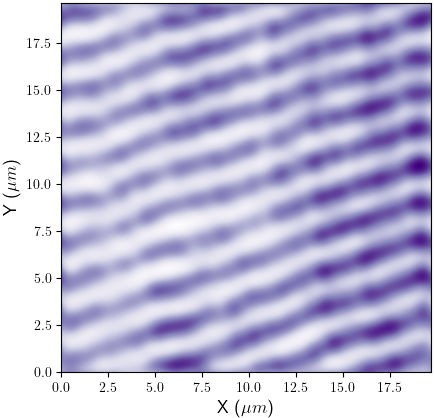} &  \includegraphics[scale=0.3]{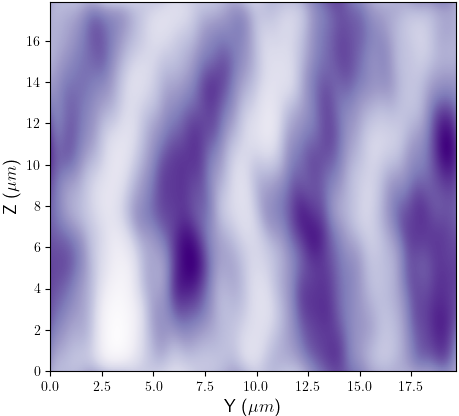} \\
           \includegraphics[scale=0.3]{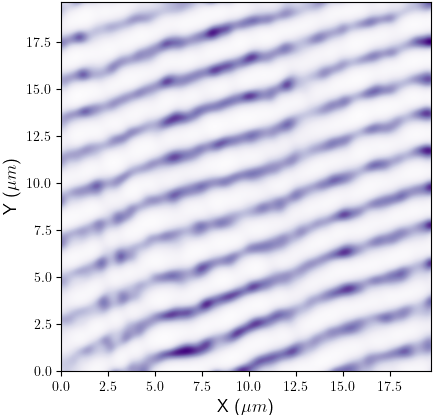} &  \includegraphics[scale=0.3]{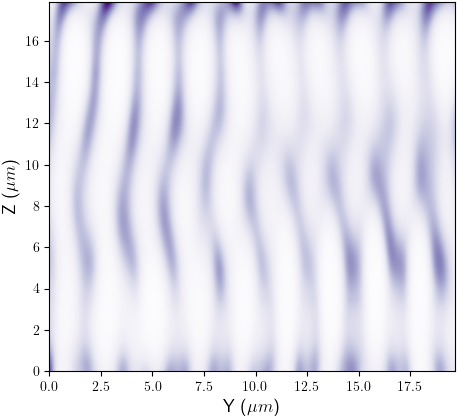} \\
           \includegraphics[scale=0.3]{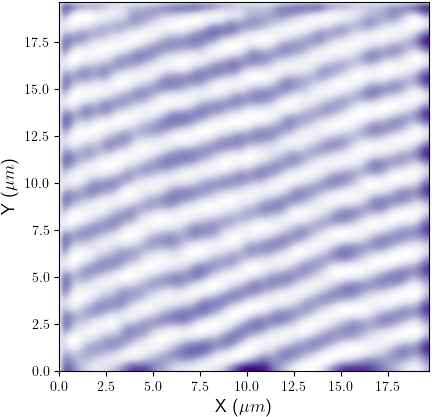} &  \includegraphics[scale=0.3]{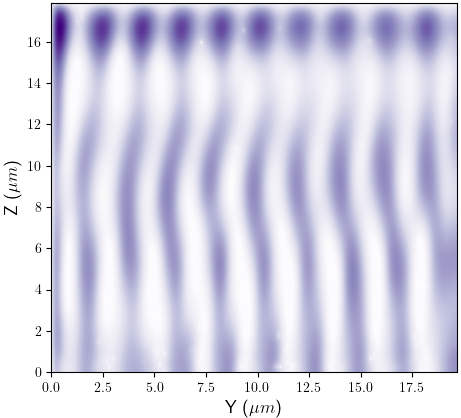} \\
            \includegraphics[scale=0.3]{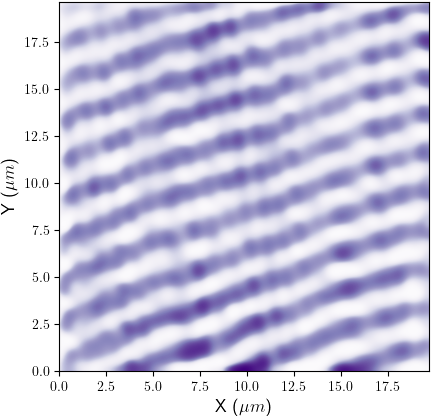} &  \includegraphics[scale=0.3]{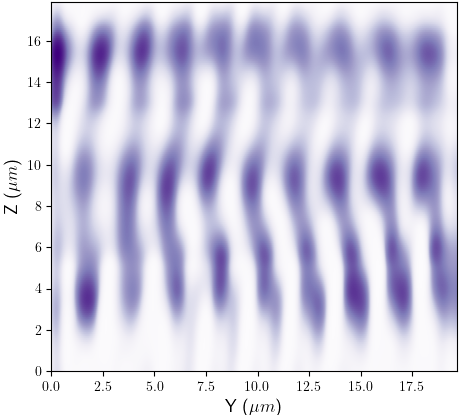} 
        \end{tabular}
      \\
    \caption{Raw image (first row) and restored image with Huygens software (second row), Krishnan et al.'s 2D blind deconvolution method \cite{krishnan2011blind} (third row), Richardson-Lucy \cite{richardson1972bayesian} (fourth row), SplitGrad (fifth row), and our proposed method GENTLE (sixth row).
    For all 3D images, we display the same slice (90th along the Z-axis) in the plane XY (first column) and the same slice (50th along the Y-axis) in the plane YZ (second column).}
    \label{fig:restoration_muscle}
\end{figure}

\begin{figure}
    \centering
    \includegraphics[scale=0.25]{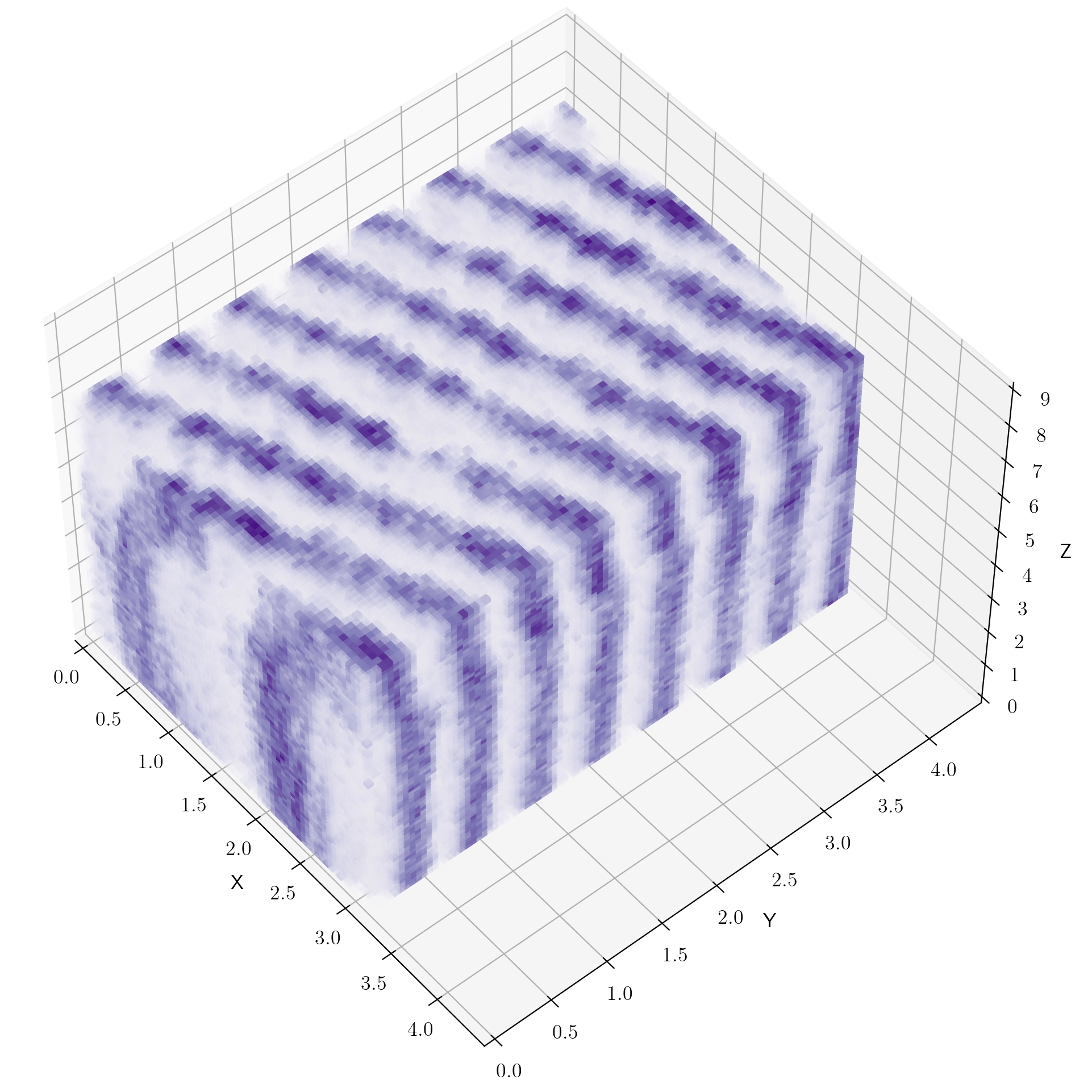}
    \includegraphics[scale=0.25]{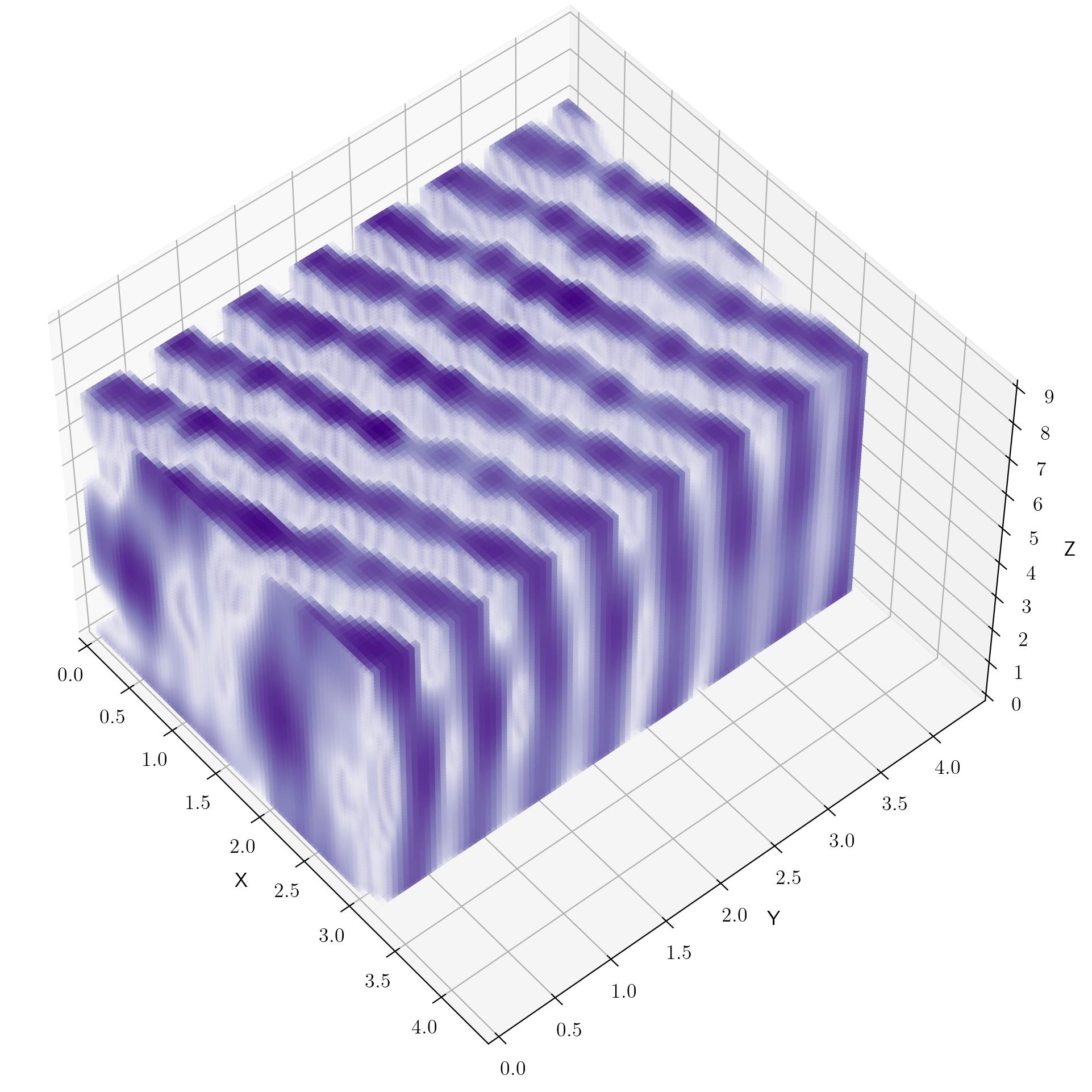}
    \caption{3D visualization of the raw volume and the restored volume using GENTLE.}
    \label{fig:3D_visualization}
\end{figure}

\begin{figure}
    \centering
    \includegraphics[scale=0.3]{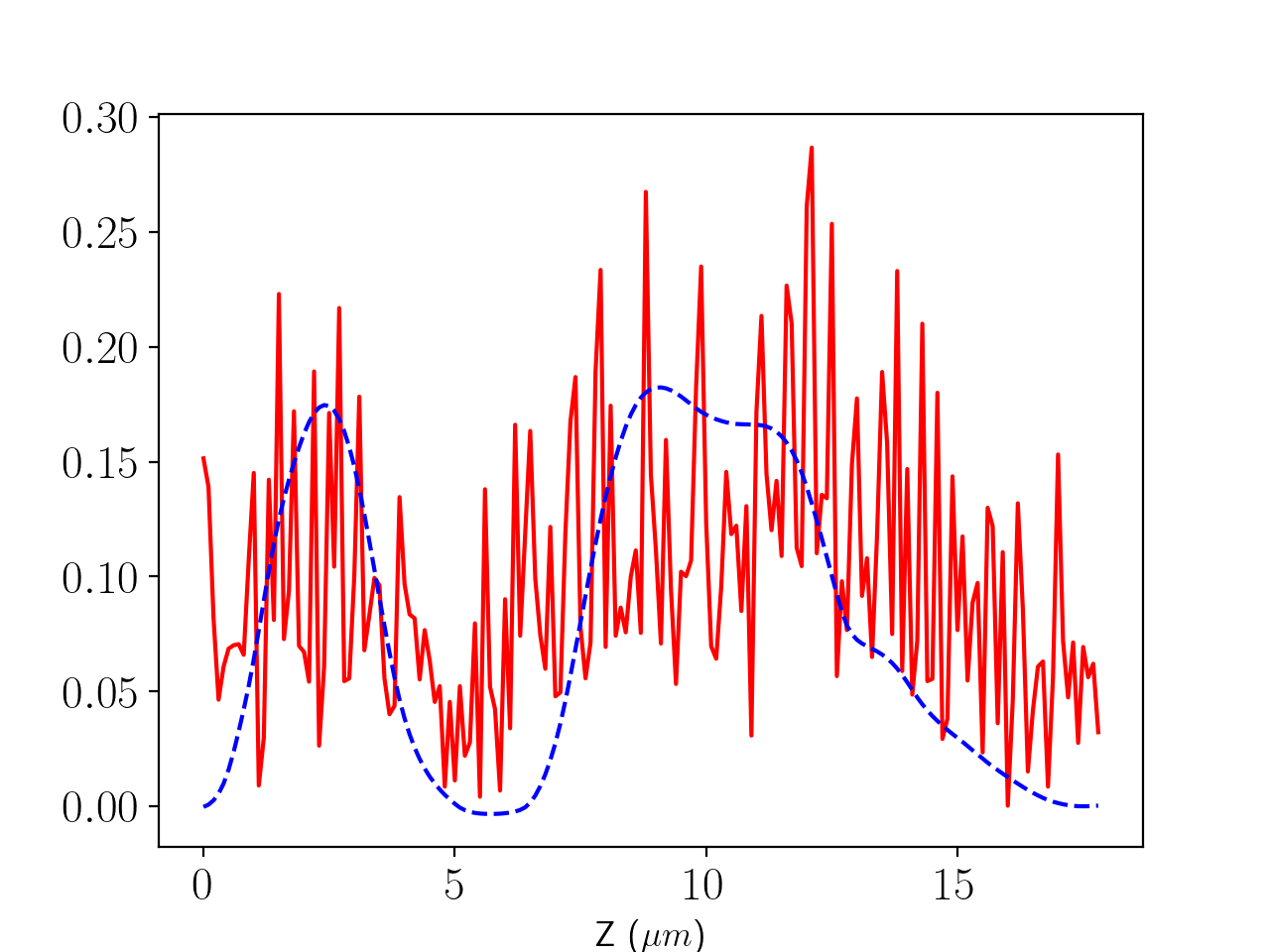}
     \includegraphics[scale=0.3]{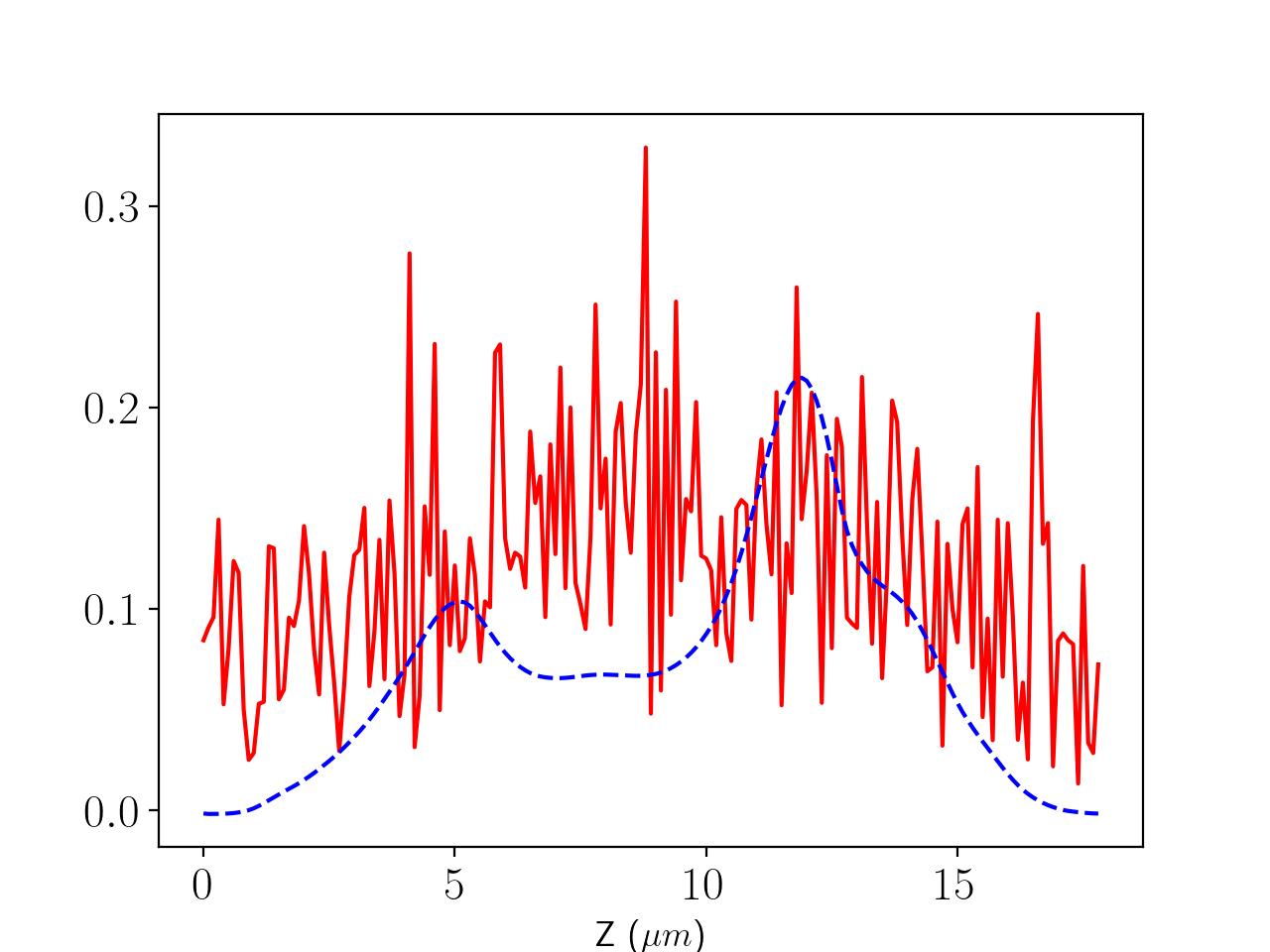}
      \includegraphics[scale=0.3]{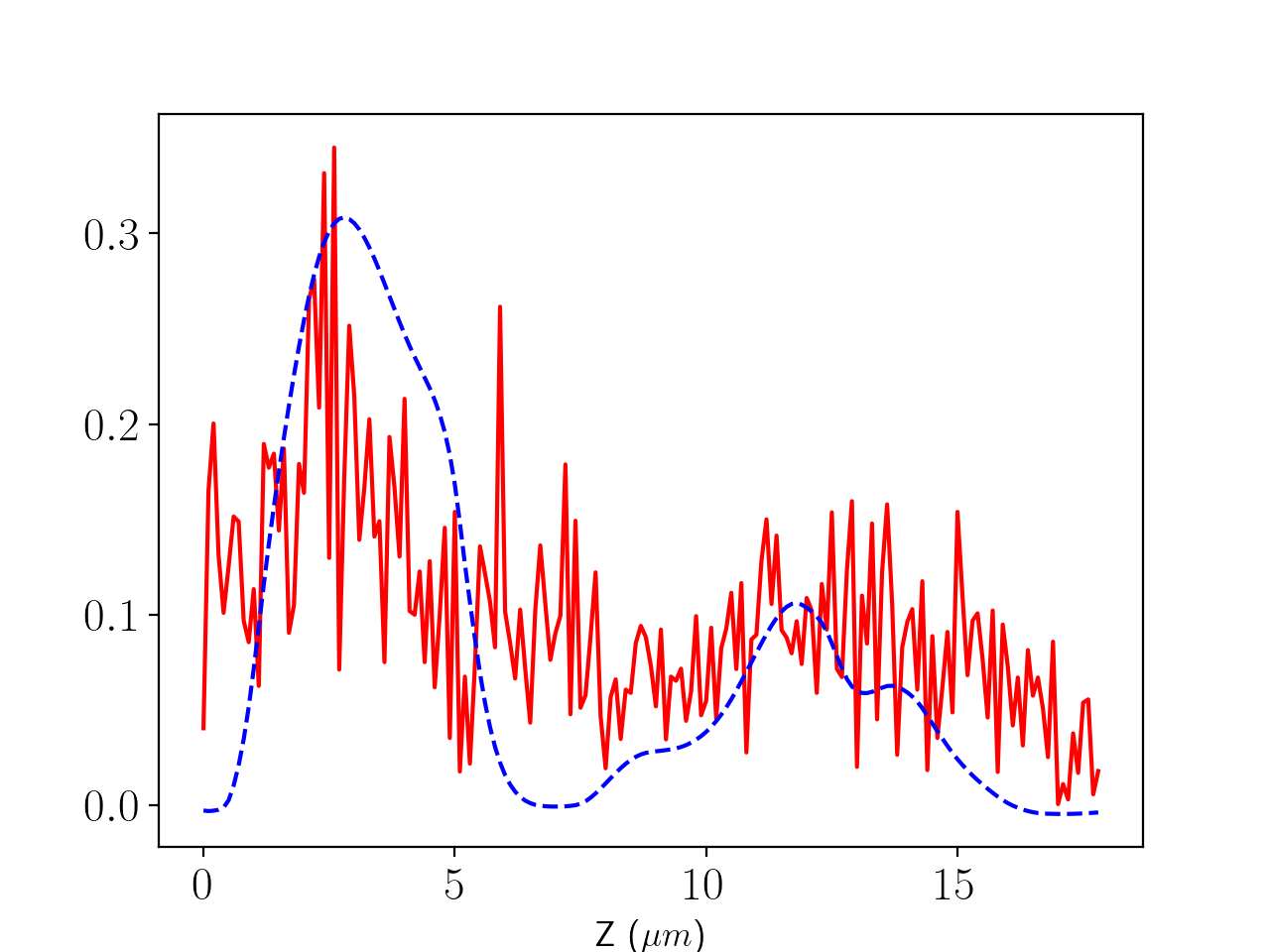}
    \caption{Profiles of the raw image (red plain line) and the restored image using GENTLE (blue dashed line) along the Z-axis, at three different positions $(X, Y)$: $(30,70)$, $(3,5)$, and $(130,25)$. }
    \label{fig:profiles_z_restoration}
\end{figure}

\section{Conclusion}

In this paper, we introduced an end-to-end approach for addressing the MPM image restoration problem. We decomposed the primary problem into two sub-problems: PSF estimation and image deconvolution, and for both, we proposed original formulation and optimization methods. Our results on simulated data, real beads imaged under optimal conditions, and mouse muscle samples, demonstrate the efficiency of our approach.

As highlighted in the paper, our method uses different noise models in its two distinct phases: initially a Gaussian model, and then a Poisson-Gaussian model. While the current setup has proven effective, considering the same heteroscedastic noise model throughout the methodology could be a promising next step to improve further the results.

Additionally, a future direction for this research could involve extending the PSF estimation strategy to other objects than beads. Given the generality of the method, it has potential applicability to any reference object with known geometrical structure.

On the restoration side, there is potential to investigate advanced regularization techniques to enhance visual outcomes as our framework is versatile, and not restricted to TV-based regularization.
Lastly, the heteroscedastic noise model we proposed for MPM deserves deeper exploration. A study analyzing how parameters $a$ and $b$ fluctuate based on different optical settings and depth in the sample could provide insights on the physical properties of the noise.

\section*{Acknowledgements}

This work was supported by the European Research Council Starting Grant MAJORIS
ERC-2019-STG-850925 and by the ANR Research and Teaching Chair BRIDGEABLE in Artificial Intelligence.

The authors would like to thank V\'eronique Blanquet, Laetitia Magnol et Alexis Parent\'e for preparing the mouse muscle samples with microbeads. 

\newpage
\bibliographystyle{abbrv}
\bibliography{biblio}

\end{document}

%% file: style.tex
\definecolor{thmcolor}{rgb}{0.7,0,0}
\definecolor{democolor}{rgb}{0.15,0.24,0.55}
\definecolor{lightgray}{rgb}{0.65,0.65,0.65}
\definecolor{hypocolor}{rgb}{0.15,0.44,0.85}
\definecolor{lightsalmon}{rgb}{1.0, 0.63, 0.48}

\newtheorem{theorem}{Theorem}

\newtheorem{proposition}{Proposition}

\theoremstyle{definition}
\newtheorem{definition}{Definition}

\LinesNumberedHidden

%% file: shortcuts.tex
\DeclareMathAlphabet{\mathb}{OML}{cmm}{b}{it}

\newcommand{\N}{\mathbb{N}}		
\renewcommand{\S}{\mathbb{S}}		
\newcommand{\R}{\mathbb{R}}		
\renewcommand{\P}{\mathbf{p}}		

\newcommand{\syst}[1]{\left \{ \begin{array}{l} #1 \end{array} \right. \kern-\nulldelimiterspace}	
\newcommand{\prox}{\text{\normalfont prox}}

\renewcommand\bs[1]{\boldsymbol{#1}}

\newfloatcommand{capbtabbox}{table}[][\FBwidth]